\documentclass[final,twoside]{IEEEtran} % regular paper
\usepackage{amssymb,amsmath,amsthm,bm,float}
\usepackage{graphicx,color}
\graphicspath{{figures-pdf/}}

\usepackage{cite}
\usepackage{url}
\usepackage{diagbox}
\usepackage[aboveskip=1pt]{subcaption}
\usepackage[mathlines]{lineno} % Line numbers on paragraphs
\usepackage{multirow,bigstrut}
\usepackage{tabularx}
\usepackage{arydshln}
\usepackage{empheq}
\usepackage{enumitem}
\usepackage{comment}% http://ctan.org/pkg/comment
\newlength\twofigwidth
\newlength\threefigwidth
\newlength\figsep
\newtheorem{Proposition}{Proposition}

\newtheorem{lemma}{Lemma}
\newtheorem{definition}{Definition}%

\newcommand{\Z}[1]{\mathbb{Z}_{p^{#1}}}
\newcommand{\Zx}[1]{\mathbb{Z}^\times_{p^{#1}}}
\newcommand{\ord}{\mathrm{ord}}
\newcommand{\logp}{\mathrm{logp}}

\newcommand{\var}[1]{\mathit{#1}}
\newcommand{\GR}{\operatorname{GR}}

\usepackage[bookmarks=false]{hyperref}
\hypersetup{
 linktocpage=true, pdfborderstyle={/S/S/W 1}, hyperindex=true, bookmarks=true, bookmarksopen=true, bookmarksnumbered=true,
}

\begin{document}

\title{Graph Structure of an Inversive Pseudorandom Number Generator over Ring $\Z{e}$}

\author{Xiaoxiong Lu, Chengqing Li, Bo Zhou
% ~\IEEEmembership{Senior Member,~IEEE,}
% <-this % stops a space
\thanks{This work was supported by the National Natural Science Foundation of China (no.~92267102).}
% Science and Technology Program of Changsha (no. kq2004021)
\thanks{X. Lu is with the School of Mathematics and Computational Science, Xiangtan University, Xiangtan, 411105, Hunan, China.}
\thanks{C. Li is with the School of Computer Science, Xiangtan University, Xiangtan, 411105, Hunan, China.}
\thanks{B. Zhou is with the School of Mathematics and Statistics, Chongqing Jiaotong University, Chongqing, 400074, China.}
}

% The paper headers
\markboth{IEEE Transactions}{Lu \MakeLowercase{et al.}}
% put a publisher's ID mark on the page
\IEEEpubid{\begin{minipage}{\textwidth}\ \\[12pt] \centering
        1549-8328 \copyright 2019 IEEE. Personal use is permitted, but republication/redistribution requires IEEE permission.\\
        See http://www.ieee.org/publications standards/publications/rights/index.html for more information.
\end{minipage}}
	
\maketitle
	
\begin{abstract}
Generating random and pseudorandom numbers with a deterministic system is a long-standing challenge in theoretical research and engineering applications. Several pseudorandom number generators based on the inversive congruential method have been designed as attractive alternatives to those based on the classical linear congruential method.
This paper discloses the least period of sequences generated by iterating an inversive pseudorandom number generator over the ring $\mathbb{Z}_e$ by transforming it into a two-order linear congruential recurrence relation. Depending on whether the sequence is periodic or ultimately periodic, all states in the domain can be attributed to two types of objects: some cycles of different lengths and one unilateral connected digraph whose structure remains unchanged concerning parameter $e$.
The graph structure of the generator over the ring $\mathbb{Z}_e$ is precisely disclosed with rigorous theoretical analysis and verified experimentally. The adopted analysis methodology can be extended to study the graph structure of other nonlinear maps.
\end{abstract}


\maketitle

\section{Introduction}

\IEEEPARstart{P}{seudorandom} number generators ($\var{PRNG}$) are deterministic algorithms designed to produce long sequences of numbers that exhibit the characteristics of randomness and are indistinguishable from true random numbers in a stream \cite{Kocarev:generators:TCASI2001,Jessa:randomBit:TC15, gongg:cycle:TIT19, xujun:congru:crypto19,licq:Logistic:IJBC2023}. These generators find wide applications in scientific and engineering domains, including Monte Carlo simulations, computer games, and cryptography \cite{cqli:block:JISA20}.
The linear congruential method, a prominent technique for generating random numbers, can be traced back to Lehmer's $\var{PRNG}$ proposed in 1951~\cite{Werner:Lehmer:NM1961}. This method operates in a multiplicative group of integers modulo $m$, with the recurrence relation $x_{n+1} = ax_n \bmod m$, where $m$ is a large integer, $x_0$ and $a$ are chosen from the set $\{1, 2, \ldots, m-1\}$, and $n$ is a non-negative integer. Subsequently, Park and Miller introduced Lehmer's $\var{PRNG}$ with specific parameters $m=2^{31}-1$ and $a=7^5$, known as the minimal standard generator (MINSTD), widely used as the default $\var{PRNG}$ in the C++ programming language~\cite{Park:goodRNG:1988ACM}. Additionally, several $\var{PRNGs}$ of the Lehmer form are integrated into the GNU Scientific Library, a numerical computing library for C and C++ programmers \cite{HandWiki:RNG}.
In 1965, Tausworthe presented a $\var{PRNG}$ that combines $L$-bit continuous binary sequences generated by the linear shift register method over a binary field~\cite{Tausworthe:LFSR:MC1965}. This $\var{PRNG}$, when compared to Lehmer's, necessitates a sufficiently large $L$ to meet the randomness requirements for numbers falling within the interval [0, 1]. Building upon Tausworthe's work, Knuth developed a generalized class of random number generators using a $k$-th order linear congruential recurrence relation over a larger finite field, extending beyond binary fields:
\begin{equation}
\label{eq:an:p}
x_n = (a_1 x_{n-1}+a_2x_{n-2}+\ldots+a_k x_{n-k}) \bmod p,
\end{equation}
where $p$ is a prime number, and parameters $a_1, a_2, \ldots, a_k$ and initial states $x_0, x_1, \ldots, x_{k-1}$ belong to the set $\{0, 1, 2, \ldots, p-1\}$~\cite{Kunth1981}. This generator encompasses Lehmer's $\var{PRNG}$ and Tausworthe's $\var{PRNG}$ as special cases. Subsequent research has explored statistical characteristics~\cite{Deng:LFSRp:1992, Fushimi:k-dis:ACM1983}, criteria for randomness assessment~\cite{Lothar:critera:JCAM1990}, expansion to other domains~\cite{Kurakin:LFSRring:JMS95, Zhu:LFSRpe:TIT2004}, and applications in cryptography based on coding theory~\cite{Li:NFSRs:TIT14}.

\IEEEpubidadjcol 

The nuanced relationship between chaotic systems and cryptosystems has positioned it as a novel avenue for designing secure and efficient pseudorandom number generators ($\var{PRNG}$)~\cite{Geros:chaoRNG:TCSI2002, Kocarev:Psebits:TCASI2003, Addabbo:Chaomap:TCSII2006}. Since von Neumann employed the Logistic map for generating random numbers in \cite{Neumann:logistic:BAMS47}, a variety of $\var{PRNGs}$ based on chaotic maps has been proposed. Examples include the Tent map~\cite{Addabbo:Tent:ITIM2006}, Logistic map~\cite{cqli:network:TCASI2019, Phatak:LogisticRNG:PRE95}, Cat map~\cite{cqli:Cat:TC21, Chen:cat:TIT2012}, Chebyshev map~\cite{Liao:Chebyshev:TC2010, Chen:ChebyshevZn:IS2011}, and R{\'e}nyi chaotic map~\cite{Addabbo:Renyi:TCASI2007}.
Notably, the Chebyshev map and Cat map are inherently linear. However, linear generators may lack cryptographic security in certain domains, prompting an intuitive preference for nonlinear alternatives~\cite{Boyar:LFSRsecret:ACM1989, Katti:CCG:TCSI2010, Panda:DualCLG:TCSI2019, xu:solving:DCC2018}.

As a promising nonlinear method for generating pseudorandom numbers, the inversive pseudorandom number generator ($\var{IPRNG}$) over a finite field $\mathbb{F}_p$ was initially proposed in \cite{Eichenauer:PSRNG:SH1986}:
\begin{equation}
\label{eq:inv:zp}
x_{n+1} = \phi_p(x_n) =
\begin{cases}
(ax_n^{-1}+b) \bmod p & \text{if } x_n \in \mathbb{Z}^{\times}_p;\\
b & \text{if } x_n = 0, 
\end{cases}
\end{equation}
where $p$ is a prime, $x^{-1}_n$ is the inverse of $x_n$ in $\mathbb{F}_p$, and $a, b \in \mathbb{F}_p$.
Leeb and Wegenkittl, through a series of empirical tests, demonstrated that $\var{IPRNG}$ can pass tests with a broader range of parameters compared to their linear counterparts~\cite{Leeb-Inver-tomacs97}. The intriguing properties of $\var{IPRNG}$ swiftly captured the attention of researchers, and it found widespread use in applications such as Quasi-Monte Carlo methods and public key encryption~\cite{Eichenauer:statinv:MC1994, Eichenauer:tomacs1992, Emmerich:aveIPNG:MC2002, xu:solving:DCC2018, Leeb-Inver-tomacs97}.
Additionally, Sol\'e \emph{et al.} extended $\var{IPRNG}$ over the Galois ring $\GR(p^e, n)$, a commutative ring with characteristic $p^e$ and $p^{en}$ elements. They posed an open problem regarding the conditions on parameters $a, b$ to achieve the maximal period~\cite{Sole:IPRG:EJC2009}. The generalization and various variants of $\var{IPRNG}$ were comprehensively reviewed in \cite{Meidl:recentwork:2010}.

Complete understanding of the period distribution and related statistical characteristics of $\var{IPRNG}$ is a basis for its extensive applications~\cite{Chou:AECC:1995, Niederreiter:Exmpe:AA2000, Niederreiter:Exm2e:IJNT2005, Eichenauer:period2e:1988, Eichenauer:periodpe:1990, Sole:IPRG:EJC2009}.
In~\cite{Chou:AECC:1995} and \cite{wsChou:IPVG:FFTA1995}, Chou described all possible periods of an $\var{IPRNG}$ over a finite field by linear difference equation and demonstrated
their relation with the period of the corresponding characteristic polynomials. 
The sufficient and necessary condition of parameters $a, b$ and initial state
to achieve the maximal period of an $\var{IPRNG}$ with modulo $2^e$ is given~\cite{Eichenauer:period2e:1988}. 
Then Eichenauer generalized the modulus in~\cite{Eichenauer:period2e:1988} to 
the power of a prime and presented the conditions that the corresponding parameters of certain specific periods meet~\cite{Eichenauer:periodpe:1990}. However, the complete information on period distribution for any parameter and initial state is not given, which limits the applications of $\var{IPRNG}$ with prime power modulus. 
Furthermore, it is noteworthy that the sequence generated by iterating $\var{IPRNG}$ may not be purely periodic, which differs from the cases of Chebyshev map and Cat map.
Especially, Chen \emph{et al.} demonstrated that if the maximal period of 
the sequence generated by the Chebyshev polynomial over finite field $\mathbb{F}_p$ is lower than $2^{64}$, the underpinned public key encryption algorithm is insecure \cite{Liao:Chebyshev:TC2010}. 
Therefore, the period analysis of the sequence generated by $\var{IPRNG}$ for all initial states is critical for any associated application.

The functional graph (also known as state-mapping network) of a map is an essential visible way to study the period distribution of sequences generated by iterating the map. The produced periodic sequence corresponds to a cycle in the corresponding functional graph.
In \cite{cqli:Cat:TC21}, Li~\emph{et al.} disclosed the graph structure of the generalized Cat map in any binary arithmetic domain and proved how the period of the sequence generated by iterating the map changes with the arithmetic precision. 
Similarly, the graph structures of the Logistic map and Tent map in a digital domain are also studied~\cite{cqli:network:TCASI2019, Yorke:tower:DCDS21}.
In the past five years, the functional graphs of various polynomials over a finite field,
Chebyshev polynomials~\cite{Daniel:Chebyshev:2018}, linearized polynomials~\cite{Daniel:lineargraph:DCC2019}, and a class of polynomials of form $x^nh(x^{\frac{q-1}{m}})$~\cite{OLi:polyField:DCC2023}, are disclosed in terms of the distribution of cycles, where $h(x)$ is a polynomial.

Although $\var{IPRNG}$ has been proposed for almost four decades, 
the structure of its functional graph remains unknown. To make the analysis 
of the graph structure more general and complete, the domain of $\var{IPRNG}$ is extended
from finite field $\mathbb{F}_p$ to ring $\Z{e}$ in this paper:
\begin{equation}
\label{eq:inv:zpe}
x_{n+1}=\phi(x_n)=
\begin{cases}
(ax_n^{-1}+b)\bmod {p^e} & \mbox{if } x_n\in\Zx{e};\\
b                        & \mbox{if } x_n \notin\Zx{e}, 
\end{cases}
\end{equation}
where $a, b, x_n\in \Z{e}$, 
$\Zx{e}=\{x \mid x\not\equiv0\pmod p, x\in\Z{e}\}$ is multiplicative group of ring $\Z{e}$
and $x_n^{-1}$ is the inverse of $x_n$ in $\Zx{e}$. 
Given an initial state $x_0\in\Z{e}$ and parameters $a$ and $b$, one can get an inversive pseudorandom number sequence ($\var{IPRNS}$) by recursion~\eqref{eq:inv:zpe}.
Then, its period is obtained using the Galois ring theory and characteristic polynomials. 
The connected component of the functional graph of $\var{IPRNG}$ over $\Z{e}$
is either a unilateral connected digraph with an invariable structure or a cycle with
no more than three possible lengths.
So, the structure of the functional graph can be unraveled by counting
the number of its different connected components.

The rest of the paper is organized as follows. 
Section~\ref{sec:pre} presents some preliminaries and lemmas that help to understand the analysis. 
Then, Sec.~\ref{sec:PeriodZp} discloses the graph structure of $\var{IPRNG}$~\eqref{eq:inv:zpe} over $\Z{}$. Furthermore, Sec.~\ref{sec:PeriodZe} presents a detailed analysis of the graph structure of the generator over Galois ring $\Z{e}$.
The last section concludes the paper.

\section{Preliminary}
\label{sec:pre}

This section introduces the relevant notations and definitions to facilitate the discussion in the ensuing sections. As for more knowledge on the finite field and Galois ring, refer to \cite{Rudolf:Introfinite:1994} and \cite{Wan:Lecture:2003}.

\subsection{Galois ring and irreducible polynomial}

In ring theory, a finite ring containing identity is called \textit{Galois ring} if its zero-divisors and zero elements form a principal ideal.
Denote $\GR(p^e, n)$ be a Galois ring of characteristic $p^e$ with cardinality $p^{en}$, where $p$ is a prime number, $e$ and $n$ are a natural number.
One can know that $\GR(p^e,1)=\Z{e}$ (See \cite[Example 14.1]{Wan:Lecture:2003}), where $\Z{e}$ is the ring of residue classes of integers modulo $p^e$ concerning modular addition and multiplication. 
Any element $\eta\in\GR(p^e, n)$ can be uniquely expressed as a $p$-adic representation
\begin{equation}
\label{eq:GR:form}
\eta=a_0+a_1p+\cdots+a_{e-1}p^{e-1}, 
\end{equation}
where $a_0, a_1, \cdots, a_{e-1}\in\Lambda$, and $\Lambda= \{0, 1, \cdots, p-1\}$ if $n=1$; $\Lambda=\{0, 1, \xi, \cdots, \xi^{p^{n}-2}\}$ otherwise, and $\xi$ is an element of order $p^n-1$ in $\GR(p^e, n)$.
All elements in $\GR(p^e, n)$ with $a_0\neq 0$ and that with $a_0=0$
form its multiplicative group $\GR^\times(p^e, n)$ and a unique maximal ideal $(p)=p\cdot \GR(p^e, n)$, respectively.
The least positive integer $m$ such that $x^m=1$ in a group is called the order of the element $x$. 
Lemma~\ref{lemma:ordk_to_e} describes the order of some special elements in $\GR^\times(p^e, n)$.

\begin{lemma}
\label{lemma:ordk_to_e}
For any $x\in\GR^\times(p^e, n)$,
\[
\ord(x)_{p^e}=
\begin{cases}
p^{e-k} & \mbox{if } x\equiv 1\pmod {p^k};\\
2p^{e-k}& \mbox{if } x\equiv-1\pmod {p^k}, 
\end{cases}
\]
where $k\geq e$ and $\ord(x)_{p^e}$ represents the order of element $x$ in $\GR^\times(p^e, n)$.
\end{lemma}
\begin{proof}
Since $x\equiv 1\pmod {p^k}$, one has $x=hp^k+1$ and
$x^m=(hp^k+1)^m=1+mhp^k+\binom{m}2h^2p^{2k}+\cdots+(hp^k)^m$, 
where $h\in \GR(p^e, n)$ and $h\not\equiv0\pmod p$. Thus, $x^m\equiv 1\pmod {p^e}$ is equivalent to
\[
mhp^k+\binom{m}2h^2p^{2k}+\cdots+(hp^k)^m \equiv 0\pmod {p^e}. 
\]
Then $m=p^{e-k}$ is the least positive integer satisfying the above congruence, which yields $\ord(x)_{p^e}=p^{e-k}$.
Similarly, when $x\equiv -1\pmod {p^k}$, $x^2 \equiv 1\pmod {p^k}$ and $\ord(x)_{p^e}=2p^{e-k}$. 
\end{proof}

Let $g(t)=a_0+a_1t+\cdots+a_nt^n$ be a polynomial over $\Z{e}$, where coefficients $a_i\in\Z{e}$, $n$ is a non-negative integer, and $t$ is an indeterminate. 
The polynomial $g(t)$ composes a polynomial ring $\Z{e}[t]$ concerning 
multiplication and addition over $\Z{e}$ \cite[Definition 1.48]{Rudolf:Introfinite:1994}.
Let $\hat{g}(t)$ denote the image of $g(t)\in\Z{e}[t]$ under a map from polynomial ring $\Z{e}[t]$ to another polynomial ring $\Z{}[t]$:
\begin{equation}    
\label{eq:map:-}
% -:\Z{e}[t] &\rightarrow \Z{}[t]\\
a_0+a_1t +\cdots +a_nt^n \mapsto \bar{a}_0+\bar{a}_1t +\cdots +\bar{a}_nt^n, 
\end{equation}
where $\bar{x}=x\bmod p$ hereinafter. 
For any polynomials $g_1(t)$, $g_2(t) \in\Z{e}[t]$, they are co-prime in $\Z{e}[t]$ if and only if $\hat{g}_1(t)$ and $\hat{g}_2(t)$ are co-prime in $\Z{}[t]$~\cite[Lemma 13.5]{Wan:Lecture:2003}.
Using such co-prime property of any two polynomials in $\Z{e}[t]$, one can prove Hensel's Lemma (Lemma~\ref{Lemma:Hensel}), which plays a vital role in the period analysis of the linear recursive sequence over ring $\Z{e}$.

\begin{lemma}
\label{Lemma:Hensel}
(Hensel's Lemma) Let $g(t)$ be a polynomial in $\Z{e}[t]$ and if
$\hat{g}(t)=h_1(t)h_2(t) \mbox{~in~} \Z{}[t]$,
where $h_1(t)$ and $h_2(t)$ are co-prime polynomials in $\Z{}[t]$. Then there exists co-prime polynomials $g_1(t)$ and $g_2(t)$ in $\Z{e}[t]$, such that
$g(t)=g_1(t)g_2(t) \text{~in~} \Z{e}[t]$, 
$\hat{g}_1(t)=h_1(t)$ and $\hat{g}_2(t)=h_2(t)$. 
\end{lemma}

A polynomial is called \textit{irreducible} over a field if it cannot be factored. In contrast, if its image under map~\eqref{eq:map:-} is irreducible in $\Z{}[t]$, the polynomial is called \textit{basic irreducible}.
Using Lemma~\ref{Lemma:Hensel}, one can get a necessary and sufficient condition that an irreducible polynomial of degree two in $\Z{e}[t]$ is basic irreducible, as shown in Lemma~\ref{lemma:ft:baisc}.

\begin{lemma}
\label{lemma:ft:baisc}
Irreducible polynomial $g(t)=t^2-bt-a$ is basic irreducible in $\Z{e}[t]$ if and only if $4a+b^2 \not\equiv 0\pmod p$. 
\end{lemma}
\begin{proof} 
When $4a+b^2 \not\equiv 0\pmod p$, assume $\hat{g}(t)$ is reducible in $\Z{}[t]$, 
one has $\hat{g}(t)=(t-t_1)(t-t_2)$, where $t_1, t_2\in \Z{}$ and $t_1\neq t_2$. 
Referring to Hensel's Lemma, one can know that there exist co-prime polynomials $g_1(t)$ and $g_2(t)$ in $\Z{e}[t]$, such that $g(t)=g_1(t)g_2(t)$ in $\Z{e}[t]$ and $\hat{g}_1(t)=t-t_1, \hat{g}_2(t)=t-t_2$, which contradicts with that $g(t)$ is irreducible.
So $g(t)$ is a basic irreducible polynomial in $\Z{e}[t]$. 
When $4a+b^2\equiv 0\pmod p$, one has $\hat{g}(t)=(t-\bar{b}/2)^2$, which yields $\hat{g}(t)$ is reducible in $\Z{}[t]$ and $g(t)$ is not a basic irreducible polynomial in $\Z{e}[t]$. 
\end{proof}

\subsection{Functional graph of $\var{IPRNS}$ over ring $\Z{e}$}

Let $S(x_0; a, b)=\{x_n\}_{n\ge 0}$ denote a specific $\var{IPRNS}$ generated by iterating $\var{IPRNG}$~\eqref{eq:inv:zpe} from an initial state $x_0\in\Z{e}$, where $x_{n}=\phi^{n}(x_0)=\phi(\phi^{n-1}(x_0))$, and $\phi^0(x_0)=x_0$. 
Its least period $T(x_0; a, b)$ is the smallest integer such that $x_{n+T(x_0; a, b)}=x_n$ for any $n\geq n_0\geq 0$.
If $n_0>0$, sequence $S(x_0; a, b)$ is \textit{ultimately periodic}, \textit{periodic} otherwise. 
To study the period of an $\var{IPRNS}$ over $\Z{e}$, 
a second-order linear recurring sequence ($\var{SLRS}$) $\{y_n\}_{n\geq 0}$ generated by 
\begin{equation}
\label{eq:yn:pe}
y_{n+2}=by_{n+1}+ay_n \bmod p^e
\end{equation}
from initial state $(y_0, y_1)=(1, x_0)$ is defined, where $a, b, x_0\in\Z{e}$. 
Then, the relation between the $\var{SLRS}$ and $\var{IPRNS}$ is revealed in Lemma~\ref{lemma:xnToyne}, which serves as the basis of the analysis of this paper. 

\begin{lemma}
\label{lemma:xnToyne}
Let $\{x_n\}_{n\geq 0}$ and $\{y_n\}_{n\geq 0}$ be sequences generated by iterating relation~\eqref{eq:inv:zpe} and relation~\eqref{eq:yn:pe} with the same parameters, respectively. Then, 
\begin{enumerate}[label=(\roman*)]
    \item If $y_n\in\Zx{e}$ for all $0\leq n\leq m$, then
    \begin{equation}\label{eq:relation:xnyn}
    x_n=y_{n+1}\cdot y^{-1}_n \bmod p^e, 
    \end{equation}
    where $m>0$ and $y^{-1}_n$ is the inverse of $y_n$ in $\Zx{e}$. 
    
    \item Integer $s$ is the smallest positive integer satisfying $x_s\in(p)$ if and only if $s+1$ is the smallest integer satisfying $y_{s+1}\in(p)$. 
\end{enumerate}
\end{lemma}
\begin{proof}
This lemma can be proved via mathematical induction $n$, and its proof is omitted. 
\end{proof}

The characteristic polynomial (or generating polynomial) of an $\var{SLRS}$ generated by relation~\eqref{eq:yn:pe} is $f(t)=t^2-bt-a$.
We call the polynomial $f(t)$ the indirect characteristic polynomial of the corresponding $\var{IPRNS}$ hereinafter.
Note that $f(t)$ is not always irreducible in $\Z{e}[t]$,
but it is always reducible in the extension ring $\mathbb{Z}_{p^e}[t]/(f(t))$, where $(f(t))=f(t)\Z{e}[t]=\{f(t)h(t): h(t)\in\Z{e}[t]\}$ is the ideal generated by $f(t)$.
Thus, we can write $f(t)$ as an uniform reducible form:
\begin{equation}\label{eq:chaacter:poly}
f(t)=(t-\alpha)(t-\beta),
\end{equation} 
where $\alpha, \beta$ belong to
\begin{equation}\label{eq:ring:Ak}
\mathbb{A}_e=
\begin{cases}
\Z{e}      & \mbox{if } f(t)\text{~is reducible in~} \Z{e}[t]; \\
\mathbb{Z}_{p^e}[t]/(f(t)) & \mbox{if } f(t)\text{~is irreducible in~} \Z{e}[t].
\end{cases}
\end{equation}
If $f(t)$ is basic irreducible in $\Z{e}[t]$,
ring $\Z{e}[t]/(f(t))$ is isomorphic to the $\GR(p^e, 2)$ \cite[Corollary 14.7]{Wan:Lecture:2003}.
Specially, $\Z{e}$ is abbreviated as $\mathbb{Z}_p$ when $e=1$.
It follows from Eq.~\eqref{eq:chaacter:poly} that
\begin{equation}
\label{eq:ab:alphabeta}
\begin{cases}
a=-\alpha\beta, \\
b=\alpha+\beta.
\end{cases}
\end{equation}

\begin{definition}\label{def:graph}
Let $\mathcal{G}(a, b, p, e)$ denote the functional graph of $\var{IPRNG}$~\eqref{eq:inv:zpe} over Galois ring $\Z{e}$.
\end{definition}

Functional graph $\mathcal{G}(a, b, p, e)$ given in Definition~\ref{def:graph} is constructed as \cite{cqli:network:TCASI2019}: the $p^{e}$ elements in ring $\Z{e}$ are considered as separate nodes; node $x$ is directly linked to node $y$ if and only if $y=\phi(x)$. 
In addition, we introduce a digraph (directed graph)
that may be connected component in $\mathcal{G}(a, b, p, e)$. 

\begin{definition}\label{def:digraph}
Let $G(L, n)$ denote a unilateral connected digraph composing a cycle of length $L$ whose sole node is linked by $n$ transient branches of length $L$ as a target. 
\end{definition}

\begin{figure}[!htb]
\centering
\begin{minipage}{0.6\twofigwidth}
\centering
\includegraphics[width=0.6\twofigwidth]{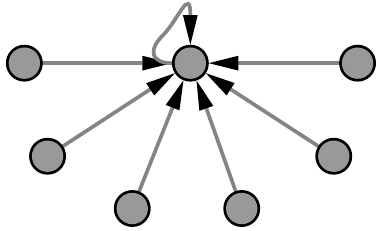}
a)
\end{minipage}
\hspace{2em}
\begin{minipage}{1\twofigwidth}
\centering 
\includegraphics[width=1\twofigwidth]{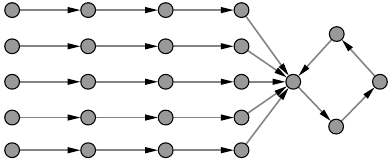}
b)
\end{minipage}
\caption{The structure of digraph $G(L, n)$: a) $G(1, 6)$; b) $G(4, 5)$.}
\label{fig:cycle:model}
\end{figure}

To facilitate understanding the definition of $G(L, n)$ given in Definition~\ref{def:digraph}, Fig.~\ref{fig:cycle:model} depicts its structure with two sets of representative parameters. 

\section{Graph structure of $\var{IPRNG}$~\eqref{eq:inv:zpe} over $\Z{}$}
\label{sec:PeriodZp}

In this section, the structure of $\mathcal{G}(a, b, p, 1)$ is disclosed via the period of every sequence generated by iterating $\var{IPRNG}$~\eqref{eq:inv:zpe}. 
Figure~\ref{fig:cycle:p1} demonstrates that with representative parameters to facilitate the description of theoretical analysis.
The analysis process of the period for two cases $ab\equiv0 \pmod p$ and $ab\not\equiv 0\pmod p$ is significantly different, so they are
separately analyzed in the following two sub-sections.

\begin{figure}[!htb]
\centering
\begin{minipage}{1\threefigwidth}
 \centering
 \includegraphics[width=0.8\threefigwidth]{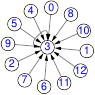}\\
 a)
\end{minipage}
\begin{minipage}{1\threefigwidth}
 \centering
 \includegraphics[width=0.97\threefigwidth]{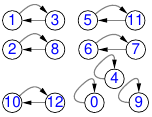}
 b)
\end{minipage}
\begin{minipage}{1\threefigwidth}
 \centering
 \includegraphics[width=0.8\threefigwidth]{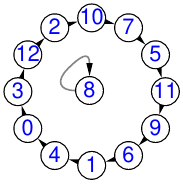}\\
 c)
\end{minipage}\\
\begin{minipage}{1\twofigwidth}
 \centering 
 \includegraphics[width=0.8\twofigwidth]{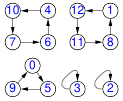}\\
 d)
\end{minipage}\hspace{2em}
\begin{minipage}{1\twofigwidth}
 \centering
 \includegraphics[width=1\twofigwidth]{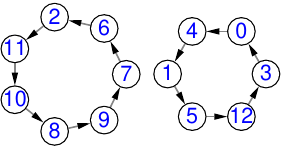}
 e)
\end{minipage}
\caption{Functional graphs of $\var{IPRNG}$~\eqref{eq:inv:zpe} over $\mathbb{Z}_{13}$: a) $(a, b)=(0, 3)$; b) $(a, b)=(3, 0)$; 
c) $(a, b)=(1, 3)$; d) $(a, b)=(7, 5)$; e) $(a, b)=(1, 4)$.}
\label{fig:cycle:p1}
\end{figure}

\subsection{Graph structure of $\var{IPRNG}$~\eqref{eq:inv:zpe} over $\Z{}$ with $ab\equiv0 \pmod p$}
\label{subsec:ab0:zp}

As there are no zero-divisors in $\Z{}$, $ab\equiv0 \pmod p$ can be divided into two cases: $a=0, b\in\Z{}$; $a\in\Zx{}, b=0$.
As for the first case, one has $x_n=b$ from $\var{IPRNG}$~\eqref{eq:inv:zpe}
for any $n\ge 1$. 
It means $S(x_0; a, b)=\{x_0, b, b, b, \cdots\}$. So, it's period $T(x_0; a, b) =1$ and 
every state evolves to state $b$.
As for the second case, one can get the concrete numbers of 
cycles of length two and self-loop composing the whole functional graph of $\var{IPRNG}$~\eqref{eq:inv:zpe}, as shown in Proposition~\ref{Th:ab0:zp}.

\begin{Proposition}\label{Th:ab0:zp}
When $a\in\Zx{}$ and $b=0$, $\mathcal{G}(a, b, p, 1)$ is composed of $\frac{p-1-|\mathbf{S}|}{2}$ cycles of length two and $|\mathbf{S}|+1$ self-loops, where $\mathbf{S}=\{x \mid x^2\equiv a\pmod p\}$ and $|\cdot|$ denotes the cardinality of a set.
\end{Proposition}
\begin{proof}
From Eq.~\eqref{eq:inv:zpe}, one has
\[
S(x_0; a, b)=
\begin{cases}
\{x_0, ax_0^{-1}, x_0, ax_0^{-1}, \cdots \} & \mbox{if } x_0\in\Zx{};\\
\{0, 0, \cdots\}       & \mbox{if } x_0=0. 
\end{cases}
\]
It means $T(x_0; a, b)=2$ if $x_0^2 \not\equiv a\pmod p$; $T(x_0; a, b)=1$ otherwise.
So, the proposition holds.
\end{proof}

As a typical example, when $(a, b, p)=(3, 0, 13)$, one has $\mathbf{S}=\{x \mid x^2\equiv 3\pmod {13}\}=\{4, 9\}$ and $|\mathbf{S}|=2$.
And the corresponding $\mathcal{G}(3, 0, 13, 1)$ shown in Fig.~\ref{fig:cycle:p1}b) is composed of $\frac{13-1-2}{2}=5$ cycles of length two and $2+1=3$ self-loops, which is consistent with Proposition~\ref{Th:ab0:zp}.

\subsection{Graph structure of $\var{IPRNG}$~\eqref{eq:inv:zpe} over $\Z{}$ with $ab\not\equiv 0\pmod p$}
\label{subsec:abn0:zp}

From Lemma~\ref{lemma:xnToyne}, one can know an $\var{IPRNS}$ generated by iterating $\var{IPRNG}$~\eqref{eq:inv:zpe} can be converted to a $\var{SLRS}$ generated by iterating relation~\eqref{eq:yn:pe} with $e=1$.
The solution of the corresponding indirect characteristic polynomial $f(t)$ has two possible cases: one root with multiplicity two; two distinct roots with multiplicity one. The general term of the associated $\var{SLRS}$ is different in the two cases, which is determined by whether congruence $4a+b^2 \equiv 0\pmod p$ exists.
Thus, the structure of $\mathcal{G}(a, b, p, 1)$ in such two cases are separately discussed in the following. 

\subsubsection{$4a+b^2\equiv 0\pmod p$}
\label{sec:root:1}
In such case, one can know $f(t)=t^2-bt-a=(t-\alpha)^2$ is reducible in $\mathbb{Z}_p[t]$, namely $\alpha=\beta\in\mathbb{A}_1=\mathbb{Z}_p$ in Eq.~\eqref{eq:chaacter:poly}.
Then,
\begin{equation}
\label{eq:ab:alpha}
\begin{cases}
a=-\alpha^2, \\
b=2\alpha.
\end{cases}
\end{equation}
Furthermore, one can get the general term of relation~\eqref{eq:yn:pe} as
\[
y_n=\alpha^n(1+c\cdot n)\bmod p, 
\]
where $c\in \mathbb{Z}_p$. Substituting $y_0=1$ and $y_1=x_0$ into the above equation yields
\begin{equation}
\label{eq:yn:root1}
y_n=\alpha^n(n(\alpha^{-1}x_0-1)+1)\bmod p.
\end{equation}
As shown in Proposition~\ref{pro:period:root2}, we can prove that the generator can produce 
maximum-length sequences (also called m-sequences).

\begin{Proposition}
\label{pro:period:root2}
When $4a+b^2\equiv 0\pmod p$, $\mathcal{G}(a, b, p, 1)$ is composed of one cycle of length $p-1$ and one self-loop. 
\end{Proposition}
\begin{proof}
When $x_0=\alpha$, one can know $x_1=\phi(x_0)=a\alpha^{-1} +b=\alpha$ from Eq.~\eqref{eq:ab:alpha}, which means sequence
$S(x_0; a, b)=\{\alpha, \alpha, \cdots\}$ and $T(x_0; a, b)=1$. 

When $x_0\neq\alpha$, there must exist an integer number $n$ such that $n(\alpha^{-1}x_0-1)+1\equiv 0 \pmod p$. 
It yields from Eq.~\eqref{eq:yn:root1} that sequence 
$\{y_n\}_{n\ge 0}$ must contain zero element. So
does sequence $S(x_0; a, b)$ from Lemma~\ref{lemma:xnToyne}-ii), which further deduces $T(x_0; a, b)=T(0;a, b)=T(b; a, b)$.
Let $x_0=b$, then $y_n=(n+1)\alpha^n\bmod p$ from Eqs.~\eqref{eq:ab:alpha} and~\eqref{eq:yn:root1}.
It means $p-1$ is the smallest integer $n$ such that $y_n=0$.
From Lemma~\ref{lemma:xnToyne}-ii), one can know $p-2$ is the smallest integer $n$ satisfying $x_n=0$. 
Then, one can get $T(b; a, b)=p-1$.
For any $x_0\in\Z{}$, there are $p-1$ initial states of value $x_0$ satisfying $x_0\neq \alpha$.
All such states make up a cycle of length $p-1$ in $\mathcal{G}(a, b, p, 1)$
with their mapping relation. 

Combining the above two cases, one can know $\mathcal{G}(a, b, p, 1)$ is composed of one cycle of length $p-1$ and one self-loop. 
\end{proof}

When $(a, b, p)=(1, 3, 13)$, one has $4a+b^2\equiv 0\pmod p$.
And the corresponding $\mathcal{G}(1, 3, 13, 1)$ shown in Fig.~\ref{fig:cycle:p1}c) is composed of one cycle of length 12 and one self-loop, which is consistent with Proposition~\ref{pro:period:root2}. 

\subsubsection{$4a+b^2\not\equiv 0\pmod p$}
\label{sec:root:2}
In such case, one has $\alpha \neq \beta$ and $\alpha,\beta \in 
\mathbb{A}_1$ from Eq.~\eqref{eq:chaacter:poly}.
Then, the general term of relation~\eqref{eq:yn:pe} is
$y_n=\lambda_1\alpha^n+\lambda_2\beta^n\bmod p$,
where $\lambda_1, \lambda_2 \in \mathbb{A}_1$. 
Substituting $y_0=1$ and $y_1=x_0$ into the previous equation, one has $\lambda_1=(x_0-\beta)(\alpha-\beta)^{-1}$ and $\lambda_2=(\alpha-x_0)(\alpha-\beta)^{-1}$. 
So, 
\begin{equation}
\label{eq:yn:root2}
y_n=(\alpha-\beta)^{-1}((x_0-\beta)\alpha^n+(\alpha-x_0)\beta^n)\bmod p. 
\end{equation}
Referring to Lemma~\ref{lemma:xnToyne},
one can know that whether $y_n=0$ for some $n$ is the key to calculate the period of sequence $S(x_0; a, b)$. 
Depending on such condition, Lemma~\ref{lemma:periodL:1root} reveals the three possible values
of the period of sequence $S(x_0; a, b)$ according to the relation between $x_0$ and $(\alpha, \beta)$.

\begin{lemma}
\label{lemma:periodL:1root}
When $a, b\in\Zx{}$ and $4a+b^2\not\equiv 0\pmod p$,
sequence $S(x_0; a, b)$ is periodic and its least period
\begin{equation*}
T(x_0; a, b)=
\begin{cases}
 1                         & \mbox{if~} x_0\in\{\alpha, \beta\}; \\
\ord(\alpha\beta^{-1})_p-1 & \mbox{if~} (x_0-\alpha)(x_0-\beta)^{-1}\in\Omega; \\
\ord(\alpha\beta^{-1})_p   & \mbox{if~} (x_0-\alpha)(x_0-\beta)^{-1}\notin\Omega,
\end{cases}
\end{equation*}
where $x_0\in\mathbb{Z}_p$ and
\[
\Omega =\{ \alpha\beta^{-1}, (\alpha\beta^{-1})^2, \cdots, (\alpha\beta^{-1})^{\ord(\alpha\beta^{-1})_p-1} \}.
\]
\end{lemma}
\begin{proof}
When $x_0\in\{\alpha, \beta\}$, one has $y_n=x^n_0$ in value from Eq.~\eqref{eq:yn:root2}. 
Equation~\eqref{eq:relation:xnyn} yields $x_n=x_0$ for any $n\geq 0$. 
Thus, sequence $S(x_0; a, b)$ is periodic and $T(x_0; a, b)=1$.

When $x_0\not\in\{\alpha, \beta\}$, one has $(x_0-\alpha)(x_0-\beta)$ is a unit,
an invertible element for the multiplication of the associated ring.
From Eq.~(\ref{eq:yn:root2}), one can obtain that $y_n=0$ if and only if
\begin{equation*}
(x_0-\alpha)\beta^n\equiv (x_0-\beta)\alpha^n \pmod p. 
\end{equation*}
Then, the above congruence is equivalent to
\begin{equation}
\label{eq:equal0:2}
(x_0-\alpha)(x_0-\beta)^{-1}\equiv(\alpha\beta^{-1})^n \pmod p.
\end{equation}
Hence, $y_n=0$ if and only if
\begin{equation}
\label{eq:set:omega}
 (x_0-\alpha)(x_0-\beta)^{-1}\in\Omega. 
\end{equation}
According to whether condition~\eqref{eq:set:omega} exists, the proof is divided into the following two cases:
\begin{itemize}
    \item $(x_0-\alpha)(x_0-\beta)^{-1}\in\Omega$: There exists a positive integer $m$ such that $y_{m}=0$ from condition~\eqref{eq:set:omega}. Then one has $x_{m-1}=0$ from Eq.~\eqref{eq:relation:xnyn}. Thus, sequence $S(x_0; a, b)$ contains zero element and $T(x_0; a, b)=T(0;a, b)=T(b; a, b)$. 
    Setting $x_0=b=\alpha+\beta$ and $y_m=0$, one has ${(\alpha\beta^{-1})}^{m+1}=1$ from Eq.~\eqref{eq:yn:root2}. 
    It means $m=\ord(\alpha\beta^{-1})_p-1$ is the smallest integer satisfying $y_{m}=0$. 
    Thus, $T(x_0; a, b)=\ord(\alpha\beta^{-1})_p-1$. Since $(b-\alpha)(b-\beta)^{-1}=\beta\alpha^{-1}=(\alpha\beta^{-1})^{\ord(\alpha\beta^{-1})_p-1}\in \Omega$, 
    and $\Omega$ is closed concerning the power of its any element, there are $\ord(\alpha\beta^{-1})_p-1$ initial states of value $x_0$ such that $(x_0-\alpha)(x_0-\beta)^{-1}\in\Omega$.
    Thus, sequence $S(x_0; a, b)$ is periodic.
    
    \item $(x_0-\alpha)(x_0-\beta)^{-1}\notin\Omega$:
    It follows that $y_n \neq 0$ for any $n$ from condition~\eqref{eq:set:omega}, 
    which yields $x_n\neq 0$ for any $n$ from Lemma~\ref{lemma:xnToyne}. 
    Then, let $x_{m}=x_0$, that is $y_{m+1}y^{-1}_{m}=x_0$, one can get $${(\alpha\beta^{-1})}^m\equiv1\pmod p$$ from Eq.~\eqref{eq:yn:root2}. 
    Thus, sequence $S(x_0; a, b)$ is periodic, and $T(x_0; a, b)=\ord(\alpha\beta^{-1})_p$ from the definition of the order of an element. 
\end{itemize}
\end{proof}

In number theory, an integer $q$ is called a quadratic residue modulo $n$ if it is congruent to a perfect square modulo $n$, that is $q \equiv x^2\pmod n$, where $x\in\{0,1,\cdots n-1\}$.
Then, if $4a+b^2$ is a quadratic residue modulo $p$, polynomial $f(t)$ is reducible over $\mathbb{Z}_p$ and $\alpha, \beta \in \Z{}$.
It follows from Lemma~\ref{lemma:periodL:1root} that $\mathcal{G}(a, b, p, 1)$ is composed of two cycles of different lengths and two self-loops, as presented in Proposition~\ref{pro:root:Zn}.
If $4a+b^2$ is a non-quadratic residue modulo $p$, $f(t)$ is irreducible over $\mathbb{Z}_p$ and $\alpha, \beta \not\in \Z{}$. So, there is no self-loop in $\mathcal{G}(a, b, p, 1)$.

\begin{Proposition}
\label{pro:root:Zn}
When $a, b\in\Zx{}$ and $4a+b^2\not\equiv 0\pmod p$, $\mathcal{G}(a, b, p, 1)$ is composed of
one cycle of length $k-1$, $\frac{p-k+1-v}{k}$ cycles of length $k$, and $v$ self-loops, 
where $k=\ord(\alpha\beta^{-1})_p$, $v=2$ if $4a+b^2$ is a quadratic residue modulo $p$; $v=0$ otherwise.
\end{Proposition}
\begin{proof}
If $4a+b^2$ is a quadratic residue modulo $p$, one has $\alpha, \beta\in \mathbb{Z}_p$.
Referring to Lemma~\ref{lemma:periodL:1root}, one can see there are three possible least periods in 
$\mathcal{G}(a, b, p, 1)$. So, there is one self-loop of length 1 corresponding to each value of $x_0$.
If $4a+b^2$ is a non-quadratic residue modulo $p$, parameters $\alpha, \beta$ belong to $GR(p,2)$ not $\Z{}$. So, the first case on $T(x_0; a, b)$ in Lemma~\ref{lemma:periodL:1root} does not exist, and there is no self-loop.
In any case, we summarize there are $v$ self-loops.
As the proof on the second case in Lemma~\ref{lemma:periodL:1root}, one can set $x_0=\alpha+\beta$ and conclude that there is only one cycle of length $\ord(\alpha\beta^{-1})_p-1$.
Excluding the $v+(\ord(\alpha\beta^{-1})_p-1)$ discussed nodes in $\mathcal{G}(a, b, p, 1)$, the remaining ones compose $\frac{p-k+1-v}{k}$ cycles of the same length, namely $k=\ord(\alpha\beta^{-1})_p$.
\end{proof}

Let $g(t)=t^2 +((a^{-1}b^2+2)\bmod p)t+1$. One can know $\alpha\beta^{-1}$ and $\alpha^{-1}\beta$ are two roots of $g(t)$~\cite[Lemma 3]{wsChou:IPVG:FFTA1995}. 
Note that $\ord(g(t))_p$ is equal to the order of any root of the polynomial~\cite[Theorem~3.3]{Rudolf:Introfinite:1994},
where $\ord(g(t))_p$ is the least positive integer $s$ satisfying $g(t)\mid t^s-1$.
Since $(\alpha\beta^{-1})^{-1}=\alpha^{-1}\beta$, 
one has $\ord(\alpha\beta^{-1})_p=\ord(\alpha^{-1}\beta)_p$.
When $4a+b^2$ is a non-quadratic residue modulo $p$, calculating the value of $\ord(\alpha\beta^{-1})_p$ is relatively complicated.
However, one can get it by
\begin{equation}\label{eq:ord:gt}
\ord(\alpha\beta^{-1})_p=\ord(g(t))_p.
\end{equation}

When $(a, b, p)=(7, 5, 13)$, one can calculate $4a+b^2\equiv1^2 \pmod{13}$, $\alpha=3$, $\beta=2$, $\beta^{-1}\equiv 7 \pmod {13}$, and $\ord(3\times7)_{13}=4$.
The corresponding $\mathcal{G}(7, 5, 13, 1)$ shown in Fig.~\ref{fig:cycle:p1}d) is composed of one cycle of length 4-1=3, $\frac{13-4+1-2}{4}=2$ cycles of length four, and two self-loops.
When $(a, b, p)=(1, 4, 13)$, one has $4a+b^2=20$ is a non-quadratic residue modulo $13$
and $g(t)=t^2+((a^{-1}b^2+2)\bmod p)t+1=t^2+5t+1$.
Assume $\epsilon$ is a root of $g(t)$, one has $\epsilon^2\equiv-5\epsilon-1\pmod {13}$. 
Enumerating the power of $\epsilon$, one can get $\epsilon^7\equiv 1\pmod {13}$. It means $\ord(\epsilon)_{13}=\ord(g(t))_{13}=7$.
The corresponding $\mathcal{G}(1, 4, 13, 1)$ shown in Fig.~\ref{fig:cycle:p1}e) is composed of one cycle of length six and one cycle of length seven. The two cases are both consistent with Proposition~\ref{pro:root:Zn}. 

\section{Graph structure of $\var{IPRNG}$~\eqref{eq:inv:zpe} over $\Z{e}$}\label{sec:PeriodZe}

Using the unilateral connected digraph $G(L, n)$, this section discloses the graph structure of $\var{IPRNG}$~\eqref{eq:inv:zpe} over $\Z{e}$ when $e\ge2$ (See some examples in Fig.~\ref{fig:cycle:p2}).
Since there are some zero-divisors in $\Z{e}$, congruence $ab\equiv 0 \pmod{p^e}$ means $a=0$; $b=0$; or $a\equiv 0\pmod{p^{e-s}}$ and $b\equiv 0\pmod{p^s}$, 
where $a, b \in \Z{e}$ and $s\in \{1, 2, \cdots, e-1\}$. 
So, solving the equation $ab=0$ in $\Z{e}$ is more complicated than that in $\Z{}$, which means 
the period analysis of $\var{IPRNS}$ is much more complex than that given in the previous section. 
Then, this section divides the whole analysis into three cases: $a\in(p)$; $a\in\Zx{e}, b\in(p)$; $a, b\in\Zx{e}$.

\begin{figure}[!htb]
\centering
\begin{minipage}{0.8\twofigwidth}
 \centering
 \includegraphics[width=0.8\twofigwidth]{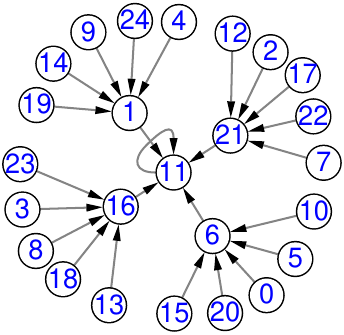}
 a) 
\end{minipage}
\hspace{1em}
\begin{minipage}{1.2\twofigwidth}
 \centering
 \includegraphics[width=1.1\twofigwidth]{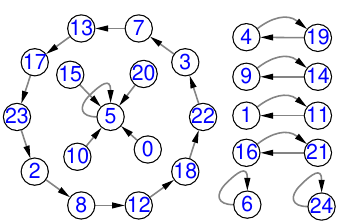}
 b)
\end{minipage}\\
\begin{minipage}{1.3\twofigwidth}
 \centering
 \includegraphics[width=1.3\twofigwidth]{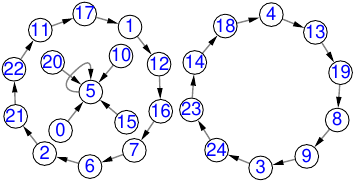}
 c)
\end{minipage}
\begin{minipage}{0.8\twofigwidth}
 \centering
 \includegraphics[width=0.8\twofigwidth]{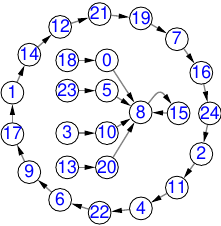}\\
 d)
\end{minipage}\\
\begin{minipage}{1.3\twofigwidth}
 \centering
 \includegraphics[width=1.3\twofigwidth]{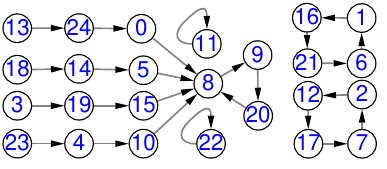}
 e)
\end{minipage}
\begin{minipage}{1.3\twofigwidth}
 \centering 
 \includegraphics[width=1.3\twofigwidth]{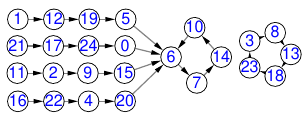}
 f)
\end{minipage}
\caption{Functional graphs of $\var{IPRNG}$~\eqref{eq:inv:zpe} over $\mathbb{Z}_{5^2}$: a) $(a, b)=(5, 6)$; b) $(a, b)=(6, 5)$; c) $(a, b)=(7, 5)$; d) $(a, b)=(6, 8)$; e) $(a, b)=(8, 8)$; f) $(a, b)=(6, 6)$.}
\label{fig:cycle:p2}
\end{figure}

\subsection{Graph structure of $\var{IPRNG}$~\eqref{eq:inv:zpe} with $a\in(p)$}
\label{subsec:ain:zpe}

When $a=0$ and $b\in\Z{e}$ or $a,b\in(p)$, one has $\phi^2(x_0)=b$ for any $x_0\in\Z{e}$, where $\phi$ is the function determined by $\var{IPRNG}$~\eqref{eq:inv:zpe}.
Thus, the corresponding functional graphs in such cases are trivial.
When $a\in(p)$, $a\neq 0$, and $b\not\in(p)$, Lemma~\ref{lemma:phir-1b} reveals that sequence $S(x_0; a, b)$ converges to a constant when its index is sufficiently large,
namely there is a self-loop in $\mathcal{G}(a, b, p, e)$.
Combining the self-loop, Lemma~\ref{lemma:condition:leafnode} gives the necessary and sufficient conditions for any node whose value is not equal to $b$ to be a non-leaf node in $\mathcal{G}(a, b, p, e)$.
Since $\phi(x_0)=b$ when $x_0\in(p)$, node of value $b$ is a non-leaf node.
Then, the in-degree of any non-leaf node is presented in Lemma~\ref{lemma:indegree:non-leafNode}.

\begin{lemma}\label{lemma:phir-1b}
When $a\in(p)$, $a\neq 0$, and $b\not\in(p)$, one has 
$\phi^l(x_0)=\phi^{l-1}(b)$ for any $x_0\in\Z{e}$, where $\phi$ is the function determined by $\var{IPRNG}$~\eqref{eq:inv:zpe}, 
$l\ge \lceil \frac{e}{\logp(a)} \rceil$, $\lceil\cdot \rceil$ is the ceiling function, and 
\[
\logp(x)=\min\{\max\{k \mid x\equiv 0 \pmod {p^k}\}, e\}.
\]
\end{lemma}
\begin{proof}
When $x_0\in (p)$, one has $\phi(x_0)=b$, which means $\phi^l(x_0)=\phi^{l-1}(b)$ for any integer $l$. So, this lemma holds for $x_0\in(p)$. 

When $x_0\not\in (p)$, by mathematical induction on integer $s$, one can prove
\begin{equation}
\label{eq:phir-1b}
\phi^s(x_0)\equiv \phi^{s-1}(b) \pmod {p^{s\cdot\logp(a)}}.
\end{equation}
When $s=1$, since $\phi(x_0)=(ax_0^{-1}+b)\bmod {p^e}$ and $a\equiv 0\pmod {p^{\logp(a)}}$, $\phi(x_0) \equiv \phi^0(b)\pmod {p^{\logp(a)}}.$
Assume that congruence~\eqref{eq:phir-1b} holds for $s=i$, namely $\phi^i(x_0)\equiv \phi^{i-1}(b) \pmod {p^{i\cdot\logp(a)}}$ and 
\begin{equation}\label{eq:phi:ix}
\phi^i(x_0)=\phi^{i-1}(b) +jp^{i\cdot\logp(a)}, 
\end{equation}
where $j$ is an integer.
It is noted that $\phi^i(x_0)$ and $\phi^{i-1}(b)$ are invertible element in $\Z{e}$ for any $i$ because $a\in(p)$ and $b\not\in(p)$. 
When $s=i+1$, from Eq.~\eqref{eq:phi:ix}, one has 
\begin{align*}
\phi^{i+1}(x_0)
& =\phi\circ\phi^i(x_0)=(a(\phi^i(x_0))^{-1}+b) \bmod p^e\\
& =(a(\phi^{i-1}(b) +jp^{i\cdot\logp(a)})^{-1}+b) \bmod p^e \\
&\equiv a(\phi^{i-1}(b))^{-1}+b \pmod{p^{(i+1)\cdot\logp(a)}} \\
&\equiv \phi^{i}(b) \pmod{p^{(i+1)\cdot\logp(a)}}.
\end{align*}
Thus, congruence~\eqref{eq:phir-1b} holds for $s=i+1$. 
The above induction completes the proof of congruence~\eqref{eq:phir-1b}. 
From congruence~\eqref{eq:phir-1b}, one can know $\phi^{l}(x_0)\equiv \phi^{l-1}(b)\pmod {p^{l\cdot\logp(a)}}$ and 
$l\cdot\logp(a)\ge e$. 
It further deduces $\phi^{l}(x_0)=\phi^{l-1}(b)$.
\end{proof}

\begin{lemma}\label{lemma:condition:leafnode}
When $a\in (p)$, $a\neq0$, and $b\not\in (p)$,
any node of value $y\neq b$ is a non-leaf node in $\mathcal{G}(a, b, p, e)$ if and only if $\logp(y-\tilde{x})\geq \logp(a)$ and $\logp(y-b)\neq \logp(a)+1$,
where $\phi(\tilde{x})=\tilde{x}$.
\end{lemma}
\begin{proof}
When a node of value $y\neq b$ is a non-leaf node in $\mathcal{G}(a, b, p, e)$. Then, there exists a node of value $x$ such that
\begin{equation}\label{eq:y:nodeAB}
y=(ax^{-1}+b)\bmod {p^e}.
\end{equation}
From the above equation and 
\begin{equation}\label{eq:x:selfloop}
\tilde{x}=\phi(\tilde{x})=(a\tilde{x}^{-1}+b) \bmod p^e, 
\end{equation}
one has 
$y-\tilde{x}= a(x^{-1}-\tilde{x}^{-1})\bmod p^e$.
It means $\logp(y-\tilde{x})\ge \logp(a)$.
Next, we prove $\logp(y-b)\neq \logp(a)+1$ via contradiction.
Since $a\in (p)$ and $a\neq0$, one can set 
\begin{equation}\label{eq:a:term}
a=\tau_a\cdot p^{\logp(a)},
\end{equation}
where $\tau_a\not\equiv 0\pmod p$ and $1\leq\logp(a)< e$.
If $\logp(y-b)=\logp(a)+1$, one has $y\equiv b\pmod {p^{\logp(a)+1}}$.
It yields from Eqs.~\eqref{eq:y:nodeAB} and~\eqref{eq:a:term} that
$(\tau_a\cdot p^{\logp(a)}x^{-1}+b)\equiv b\pmod {p^{\logp(a)+1}}$.
It further deduces $\tau_a \equiv 0 \pmod{p}$, which contradicts with the definition of $\tau_a$. 
So, $\logp(y-b)\neq \logp(a)+1$. 

When a node whose value $y$ satisfies $\logp(y-\tilde{x})\geq \logp(a)$ and $\logp(y-b)\neq\logp(a)+1$, 
one has $y= \tilde{x}+jp^{\logp(a)}$ and $\tilde{x}+jp^{\logp(a)}-b\not\equiv 0\pmod {p^{\logp(a)+1}}$, where $j$ is an integer.
It yields from Eqs.~\eqref{eq:x:selfloop} and \eqref{eq:a:term} that 
\[
\tau_ap^{\logp(a)}(\tilde{x}^{-1}+\tau_a^{-1}j) \not\equiv  0\pmod {p^{\logp(a)+1}}.
\]
Thus, $\tilde{x}^{-1}+\tau_a^{-1}j$ is invertible.
When $x=(\tilde{x}^{-1}+\tau_a^{-1}j)^{-1}$, one can get 
\begin{align*}
\phi(x)& = a(\tilde{x}^{-1}+\tau_a^{-1}j)+b\bmod {p^e}\\
& = a\tilde{x}^{-1}+b+jp^{\logp(a)} \bmod {p^e}.
\end{align*}
Substituting Eq.~\eqref{eq:x:selfloop} and $y-\tilde{x}=jp^{\logp(a)}$ into the above equation, 
one has $\phi(x) =y$. 
Thus, the node is a non-leaf node in $\mathcal{G}(a, b, p, e)$.
\end{proof}

\begin{lemma}\label{lemma:indegree:non-leafNode}
When $a\in (p)$, $a\neq0$, and $b\not\in (p)$, the in-degree of any non-leaf node of state value $y$ in $\mathcal{G}(a, b, p, e)$ is $p^{\logp(a)}$ if $y\neq b$; $p^{e-1}$ otherwise.
\end{lemma}
\begin{proof}
If $y=b$, one has $\phi(x)=b$ for any $x\in(p)$ from Eq.~\eqref{eq:inv:zpe}. 
So, the in-degree of the non-leaf node in $\mathcal{G}(a, b, p, e)$ is $|(p)|=p^{e-1}$.

If $y\neq b$, one can know there exists $x\in\Zx{e}$ such that
$y=\phi(x)=(ax^{-1}+b) \bmod{p^e}$.
According to Eq.~\eqref{eq:GR:form}, one can set 
% $x=\sum_{i=0}^{e-1}c_ip^{i}$ and
$x^{-1} =\sum_{i=0}^{e-1}d_ip^{i}$,
where $d_i\in\{0,1,\cdots, p-1\}$ and $d_0\neq 0$.
Then, from Eq.~\eqref{eq:a:term}, one has 
\begin{align*}
y & = \left(\tau_a\cdot\sum_{i=0}^{e-1}d_ip^{i+\logp(a)} +b \right)\bmod{p^e}\\
  & = \left(\tau_a\cdot\sum_{i=0}^{e-\logp(a)-1}d_ip^{i+\logp(a)}+b\right) \bmod{p^e},
\end{align*}
which holds for any $i\in\{e-\logp(a),e-\logp(a)+1, \cdots,e-1\}$, and any $d_i\in \{0,1,\cdots, p-1\}$,
namely there are $p^{\logp(a)}$ states of value $x$ satisfying the above equation.
So, the in-degree of the non-leaf node in $\mathcal{G}(a, b, p, e)$ is $p^{\logp(a)}$.
\end{proof}

As shown in Fig.~\ref{fig:cycle:p2}a), when $(a, b, p, e)=(5, 6, 5, 2)$, one can know $\logp(a)=1$ and the self-loop $\tilde{x}=11$.
For edges $``3\rightarrow16"$ and $``13\rightarrow16"$, one can calculate $\logp(x-11)=0$ when $x=3,13$ and $\logp(16-11)=1$. So, there is a rule $\logp(\phi(x)-\tilde{x})=\logp(x-\tilde{x})+ \logp(a)$ in such case.
Then, the similar rule between the self-loop in $\mathcal{G}(a, b, p, e)$ and any node and its linking node is presented in Lemma~\ref{lemma:nodeATonodeB}.
Using the Lemma, Proposition~\ref{Th:a0:Zpe} discloses how any node evolves to the self-loop.

\begin{lemma}\label{lemma:nodeATonodeB}
When $a\in (p)$, $a\neq0$, and $b\not\in (p)$, one has 
\begin{equation}\label{eq:phi(x):self}
\logp(\phi(x)- \tilde{x})= \min\{\logp(a)+\logp(x-\tilde{x}), e\},
\end{equation}
where $\phi(\tilde{x})=\tilde{x}$.
\end{lemma}

\begin{proof}
Let $t=\logp(x-\tilde{x})$, one has 
$x=\tilde{x}+jp^t$ and $j\not\equiv0\pmod p$. 
Then, one can get
\begin{equation}\label{eq:selfloop:ite}
\phi(x)= 
\begin{cases}
(a(\tilde{x}+jp^t)^{-1}+b)\bmod {p^e} & \mbox{if } x\in\Zx{e};\\
b      & \mbox{if } x\notin\Zx{e}.
\end{cases}
\end{equation}
For the first case of Eq.~\eqref{eq:selfloop:ite}, 
one can obtain
\begin{align*}
\phi(x)
& = ( a(\tilde{x}^{-1}-jp^t\tilde{x}^{-1}(\tilde{x}+jp^t)^{-1})+b) \bmod{p^e}\\ 
& = (a\tilde{x}^{-1}+b-ajp^t\tilde{x}^{-1}(\tilde{x}+jp^t)^{-1}) \bmod{p^e}.
\end{align*}
It yields from Eqs.~\eqref{eq:x:selfloop} and~\eqref{eq:a:term} that
\[
\phi(x)=(\tilde{x}- p^{\logp(a)+t}\cdot\tau_a\cdot j\tilde{x}^{-1}(\tilde{x}+jp^t)^{-1}) \bmod{p^e}.
\]
Then, referring to
$\tau_a\cdot j\tilde{x}^{-1}(\tilde{x}+jp^t)^{-1})\not\equiv0\pmod p$, one can get 
\[
\phi(x)\equiv 
\begin{cases}
\tilde{x}\pmod {p^{\logp(a)+t}}  & \mbox{if } 0\le t\le e-\logp(a);\\
\tilde{x}\pmod {p^e} & \mbox{otherwise.}
\end{cases}
\]
Namely, Eq.~\eqref{eq:phi(x):self} holds.

For the second case of Eq.~\eqref{eq:selfloop:ite}, one has $t=0$ from $x=\tilde{x}+jp^t$.
According to Eqs.~\eqref{eq:x:selfloop} and~\eqref{eq:a:term}, one has $\phi(x)-\tilde{x}\equiv -p^{\logp(a)}\cdot\tau_a\tilde{x}^{-1}\pmod{p^e}$,
namely, $\logp(\phi(x)-\tilde{x})=\logp(a)$.
So, Eq.~\eqref{eq:phi(x):self} holds.
\end{proof}

\begin{Proposition}
\label{Th:a0:Zpe}
When $a\in (p)$, $a\neq0$, and $b\not\in (p)$, 
any node of value $x$ evolves to a unique state of value $\tilde{x}=h-j$ via a shortest path of length $t(x)=h-j$ in $\mathcal{G}(a, b, p, e)$, where
$h=\lceil \frac{e}{\logp(a)}\rceil-1$,
$j\in \{-1,0, \cdots, h-1\}$ satisfies $\logp(x-\tilde{x})\in T_j$,
$T_{-1}=\{0,1,\cdots, k-1\}$,
$T_j=\{j\cdot \logp(a)+k,j\cdot \logp(a)+k+1, \cdots,(j+1)\cdot \logp(a)+k-1\}$ when $j\neq -1$, 
$k=e-h\logp(a)$, and $\tilde{x}=\phi^{h}(b)$. 
\end{Proposition}
\begin{proof}
Setting $l=h+1$ and $x_0=\phi^{h}(b)$ in Lemma~\ref{lemma:phir-1b}, one has $\phi(\phi^{h}(b))=\phi^{h}(b)$, namely $\tilde{x}=\phi^{h}(b)$.
From the definition of $t(x)$, one has $t(x)$ is the least non-negative integer such that
$\phi^{t(x)}(x)=\tilde{x}$.
Referring to Lemma~\ref{lemma:nodeATonodeB}, 
one can prove
\begin{equation}\label{eq:phis:x}
\logp(\phi^s(x)- \tilde{x})=\min\{s\cdot\logp(a)+\logp(x-\tilde{x}), e\},
\end{equation}
where $s$ is an integer.
Thus, one can know the value of $t(x)$ depends on the value of $\logp(x-\tilde{x})$, where $\logp(x-\tilde{x})\in\{0,1,\cdots e\}$.
Combining a partition of the set
% https://en.wikipedia.org/wiki/Partition_of_a_set
\[
\{0,1,\cdots, e\}=\cup_{j=-1}^{h-1}T_j\cup\{e\},
\]
the proof is divided into the following two cases:
\begin{itemize}
    \item $\logp(x-\tilde{x})=e$: The node of value $x=\tilde{x}$ composes a self-loop.
    \item $\logp(x-\tilde{x})\in T_j$: It yields from Eq.~\eqref{eq:phis:x} that 
    \begin{multline*}
    \logp(\phi^{h-1-j}(x)-\tilde{x})=\\\min\{(h-1-j)\logp(a)+\logp(x-\tilde{x}),e \}.       
    \end{multline*}
    From the definition of set $T_j$ and $k=e-h\logp(a)$, one can get
    \[\logp(\phi^{h-1-j}(x)-\tilde{x})\leq h\logp(a)+k-1=e-1.\]
    It means $\phi^{h-1-j}(x)\neq\tilde{x}$.
    Then, it follows from Eq.~\eqref{eq:phis:x} that 
    $\logp(\phi^{h-j}(x)-\tilde{x})=\min\{h', e\},$
    where $h'=(h-j)\logp(a)+\logp(x-\tilde{x})$.
    Note that 
    \begin{align*}
    h' &\ge 
    \begin{cases}
     (h+1)\logp(a)  & \mbox{if } j=-1;\\
     h\logp(a)+k  & \mbox{otherwise;}
    \end{cases}\\\
    &\ge e.
    \end{align*}
    It means $\logp(\phi^{h-j}(x)-\tilde{x})=e$, namely $\phi^{h-j}(x)=\tilde{x}$.
    So, $t(x)=h-j$.
\end{itemize}
\end{proof}

\iffalse
As shown in Fig.~\ref{fig:cycle:p2}a), when $(a, b, p, e)=(5, 6, 5, 2)$, there is a unique self-loop ``11" in $\mathcal{G}(5, 6, 5, 2)$ and every node evolves to the loop via a path, which is consistent with Proposition~\ref{Th:a0:Zpe}.
\fi

\subsection{Graph structure of $\var{IPRNG}$~\eqref{eq:inv:zpe} with $a\in\Zx{e}$ and $b\in(p)$}\label{subsec:bin:zpe}

As $\var{IPRNG}$~\eqref{eq:inv:zpe} maps any initial state $x_0\not\in\Zx{e}$ 
into a fixed value $b$, the structure of the associated connected component is trivial and known.
As for any initial state belonging to $\Zx{e}$, $\var{IPRNG}$~\eqref{eq:inv:zpe} composes a permutation when
$a\in\Zx{e}$ and $b\in(p)$ \cite{Niederreiter:Exmpe:AA2000}, which means sequence $S(x_0; a, b)$ is periodic.
So, the graph structure of $\var{IPRNG}$~\eqref{eq:inv:zpe} can be disclosed by studying the period of every sequence starting from
every initial state $x_0\in \Zx{e}$.

From Eq.~\eqref{eq:ab:alphabeta}, one has $a=-\alpha\beta\in\Zx{e}$ and $b=\alpha+\beta\in(p)$.
It yields that $\beta\notin(p)$ and
\begin{equation}\label{eq:albeta:no0}
\alpha-\beta=b-2\beta \not\in(p).
\end{equation}
Similar to Eq.~\eqref{eq:yn:root2}, one has
the general term of relation~\eqref{eq:yn:pe} as
\begin{equation}\label{eq:ynroot2:pe}
y_n=(\alpha-\beta)^{-1}((x_0-\beta)\alpha^n+(\alpha-x_0)\beta^n)\bmod {p^e}.
\end{equation}
For any $x_0\in\Zx{e}$, let $x_{m}=x_0$, that is $y_{m+1}\cdot y_m^{-1}=x_0$ from Eq.~\eqref{eq:relation:xnyn}. 
It follows from Eq.~\eqref{eq:ynroot2:pe} that 
\begin{equation}
\label{eq:alpha:zpe}
(x_0-\alpha)(x_0-\beta)\alpha^m \equiv (x_0-\alpha)(x_0-\beta)\beta^m \pmod {p^e}.
\end{equation}
Thus, the period of the sequence $S(x_0; a, b)$ equals the smallest integer $m$ satisfying congruence~\eqref{eq:alpha:zpe}. 
Furthermore, Lemma~\ref{lemma:L:an0b0} discloses 
two possible cases on the period of sequence $S(x_0; a, b)$, which is determined by the relation between $x_0$ and $(\alpha, \beta)$.

\begin{lemma}
\label{lemma:L:an0b0}
When $a\in\Zx{e}$ and $b\in(p)$, the least period of sequence $S(x_0; a, b)$ is 
\begin{equation}
\label{eq:L:an0b0}
T(x_0; a, b)=
\begin{cases}
1 & \mbox{if } x_0\in\{\alpha, \beta\}; \\
\ord(\alpha\beta^{-1})_{p^{e-k'}},
& \mbox{if } x_0\not\in\{\alpha, \beta\},
\end{cases}
\end{equation}
where $x_0\in\Zx{e}$, 
$k'=\max(k_{\alpha}, k_{\beta})$,
$k_{\alpha}=\logp(x_0-\alpha)$, and
$k_{\beta}=\logp(x_0-\beta)$.
\end{lemma}

\begin{proof}
If $x_0\in\{\alpha, \beta\}$, one has $y_n=x_0^n\bmod p^e$ from Eq.~\eqref{eq:ynroot2:pe}.
Thus, $x_n=y_{n+1}\cdot y_n^{-1}= x_0$ for any $n\geq0$ and $T(x_0; a, b)=1$. 
If $x_0\not\in\{\alpha, \beta\}$, the proof of this lemma is divided into the following two cases:
\begin{itemize}
    \item $(x_0-\alpha)(x_0-\beta)\not\in(p)$: 
    In such case, $\max(k_{\alpha}, k_{\beta})=0$ and congruence~\eqref{eq:alpha:zpe} is equivalent to $(\alpha\beta^{-1})^m\equiv 1\pmod {p^{e}}$. 
    Thus, $\ord(\alpha\beta^{-1})_{p^{e}}$ is the smallest integer satisfying congruence~\eqref{eq:alpha:zpe} and $T(x_0; a,  b)=\ord(\alpha\beta^{-1})_{p^{e}}=\ord(\alpha\beta^{-1})_{p^{e-k'}}$. 
    
    \item $(x_0-\alpha)(x_0-\beta)\in(p)$: If $x_0-\alpha, x_0-\beta \in (p)$,
    one has $\alpha -\beta \in(p)$. It contradicts Eq.~\eqref{eq:albeta:no0}.
    So, such case includes two possible sub-cases:
    $x_0-\alpha \in(p)$ and $x_0-\beta \not\in(p)$; $x_0-\alpha \not\in(p)$ and $x_0-\beta \in(p)$.
    For the first sub-case, congruence~\eqref{eq:alpha:zpe} is equivalent to
    \[(\alpha\beta^{-1})^m\equiv 1\pmod {p^{e-k_\alpha}}.\]
    So, $\ord(\alpha\beta^{-1})_{p^{e-k_\alpha}}$ is the smallest integer $m$ satisfying the above congruence, which means $T(x_0; a, b)=\ord(\alpha\beta^{-1})_{p^{e-k_\alpha}}$. 
    The proof of the second sub-case is similar and omitted.
\end{itemize}
\end{proof}

From Eqs.~\eqref{eq:ab:alphabeta} and \eqref{eq:albeta:no0}, one has
$4a+b^2 =(\alpha-\beta)^2\not\equiv 0\pmod p$.
If $4a+b^2$ is a quadratic residue modulo $p^e$, $f(t)$ is reducible over $\Z{e}$ and $\alpha, \beta \in \Z{e}$. 
It follows from Lemma~\ref{lemma:L:an0b0} that there are 
more than two cases on the lengths of cycles in $\mathcal{G}(a, b, p, e)$ for any $e$.
As shown in Table~\ref{tab:ab=6_25_cycle}, the number of cycles of length $T_c$ remains unchanged with the increase of $e$. 
Furthermore, Proposition~\ref{pro:fzpe:unit} describes the structure of $\mathcal{G}(a, b, p, e)$ in such case.
If $4a+b^2$ is a non-quadratic residue modulo $p^e$, $f(t)$ is irreducible over $\Z{e}$ and $\alpha, \beta \not\in \Z{e}$, and there is no self-loop in $\mathcal{G}(a,b, p, e)$, as shown in Proposition~\ref{pro:abn0b0:fnot}. 

\begin{table}[!htb]
\centering
\caption{The number of cycles of length $T$, $N_{T}$, in $\mathcal{G}(a, b, p, e)$ with $(a, b)=(6, 25)$.}
\begin{tabular}{*{7}{c|}c}
% {*{10}{c|}c}
\hline
\diagbox[width=6em]{$e$}{$N_{T}$}{$T$}
 &1	&2	&10	&$2\cdot5^2$ &$2\cdot5^3$ &$2\cdot5^4$	&$2\cdot5^5$ \\ \hline
1	&2	&1	&	&	&	&	&	\\ \hline
2	&2	&9	&	&	&	& &	\\ \hline
3	&2	&24	&5	&	&	&	&	\\ \hline
4	&2	&24	&20	&5	&	&	&	\\ \hline
5	&2	&24	&20	&20	&5	&	&	\\ \hline
6	&2	&24	&20	&20	&20	&5	&	\\ \hline
7	&2	&24	&20	&20	&20	&20	&5	\\ \hline
% 8	&2	&24	&20	&20	&20	&20	&20 &5 \\ \hline
\end{tabular}
\label{tab:ab=6_25_cycle}
\end{table}

\begin{Proposition}
\label{pro:fzpe:unit}
When $a\in\Zx{e}$, $b\in(p)$, and $4a+b^2$ is a quadratic residue modulo $p^e$, $\mathcal{G}(a, b, p, e)$ is composed of
connected component $G(1, p^{e-1}-1)$, $\frac{(p-3)p^{k_b-1}}{2}$ cycles of length $2p^{e-k_b}$, $w$ cycles of length $2p^{e-k_b-k'}$, $v$ cycles of length two,
and two self-loops, 
where $k_b=\logp(b)$, $k'\in\{1,2,\cdots, e-k_b-1\}$, and
\begin{equation*}
(v, w)=
\begin{cases}
(p^{k_b}-1, (p-1)p^{k_b-1}) & \mbox{if } e\ge k_b+2; \\
(p^{e-1}-1, 0)              & \mbox{otherwise}.
\end{cases} 
\end{equation*}
\end{Proposition}
\begin{proof}
When $x_0\in (p)$, one has $S(x_0; a, b)=\{x_0, b, b, \cdots\}$ from $\var{IPRNG}$~\eqref{eq:inv:zpe}. So, all states whose values belong to $(p)$ compose connected component $G(1, p^{e-1}-1)$. 

When $x_0\not\in(p)$, the related analysis is divided into the following two cases:
\begin{itemize}
    \item $(x_0-\alpha)(x_0-\beta)\notin(p)$: Referring to Lemma~\ref{lemma:L:an0b0}, one has
    $T(x_0; a, b)=\ord(\alpha\beta^{-1})_{p^e}. $
    Since $b\in(p)$ and the definition of $k_b$, 
    $b=\alpha+\beta\equiv 0\pmod {p^{k_b}}$ from Eq.~\eqref{eq:ab:alphabeta}, which further deduces $\alpha\beta^{-1}+1\equiv 0 \pmod {p^{k_b}}$. 
    It yields from Lemma~\ref{lemma:ordk_to_e} that 
    $\ord(\alpha\beta^{-1})_{p^e}=2p^{e-k_b}.$
    So, $T(x_0; a, b)=2p^{e-k_b}.$
    From Eq.~\eqref{eq:GR:form}, one can get the $p$-adic representation of $\alpha, \beta, x_0$:
    \begin{equation}\label{eq:p-adic}
    \begin{cases}
    \alpha=\sum^{e-1}_{i=0}c_ip^i;\\
    \beta=\sum^{e-1}_{i=0}d_ip^i;\\
     x_0=\sum^{e-1}_{i=0}h_ip^i, 
    \end{cases}
    \end{equation}
    where $c_0, d_0, h_0\in \{1, 2, \cdots, p-1\}$ and $c_i, d_i, h_i\in \{0, 1, \cdots, p-1\}$ for any $i \in \{1, 2, \cdots, e-1\}$. 
    Since $x_0-\alpha, x_0-\beta, x_0\in\Zx{e}$, 
    one has $h_0\notin \{0, c_0, d_0\}$ from Eq.~\eqref{eq:p-adic}. 
    Thus, there are $(p-3)p^{e-1}$ initial states of value $x_0$ such that 
    $(x_0-\alpha)(x_0-\beta)\not\in(p)$, and these states compose $\frac{(p-3)p^{e-1}}{2p^{e-k_b}}=\frac{(p-3)p^{k_b-1}}{2}$ cycles of length $2p^{e-k_b}$ in $\mathcal{G}(a, b, p, e)$.
    
    \item 
    $(x_0-\alpha)(x_0-\beta)\in(p)$: One can assume $x_0-\alpha\in(p)$ and $x_0-\beta\not\in (p)$.
    The proof of the case $x_0-\alpha\not\in(p)$ and $x_0-\beta\in (p)$ is similar and omitted.
    When $x_0=\alpha$, one has $T(x_0; a, b)=1$ from Lemma~\ref{lemma:L:an0b0}. 
    It means the state of value $\alpha$ composes a self-loop in $\mathcal{G}(a, b, p, e)$. 
    When $x_0\neq\alpha$, from the definitions of $k_\alpha$ and $k_\beta$, one has $1\leq k_\alpha\leq e-1$ and $k_\beta=0$.
    It further deduces 
    $T(x_0; a, b)=\ord(\alpha\beta^{-1})_{p^{e-k_\alpha}}$
    from Lemma~\ref{lemma:L:an0b0}. 
    Since $b=\alpha+\beta\in(p)$, one has $\alpha\beta^{-1}+1 \equiv 0 \pmod {p^{k_b}}$ and 
    \[
    \alpha\beta^{-1} \equiv
    \begin{cases}
    -1\pmod {p^{e-k_\alpha}} & \mbox{if } e-k_\alpha\leq k_b \leq e;\\
    -1\pmod {p^{k_b}}  & \mbox{if } 1\leq k_b<e-k_\alpha. 
    \end{cases}
    \]
    According to Lemma~\ref{lemma:ordk_to_e} and the above congruence, one can get
    \begin{align}
    T(x_0; a, b) & =\ord(\alpha\beta^{-1})_{p^{e-k_\alpha}} \nonumber \\
                 & =                                  \label{eq:T:k_a}
    \begin{cases}
    2p^{e-k_\alpha-k_b} & \mbox{if } 1\le k_\alpha< e-k_b;\\
    2                   & \mbox{if } e-k_b\le k_\alpha < e.
    \end{cases}
    \end{align}
    If $k_b \in \{e-1, e\}$, only the second case of Eq.~\eqref{eq:T:k_a} holds, which means $T(x_0; a, b)=2$. 
    Referring to Eq.~\eqref{eq:p-adic}, one can get there are $p^{e-1}-1$ initial states of value $x_0$ satisfying
    $x_0-\alpha \in(p)$ and $x_0\neq\alpha$.
    Thus, all states of value $x_0$ satisfying $(x_0-\alpha)(x_0-\beta)\in(p)$ compose $p^{e-1}-1$ cycles of length two and two self-loops in $\mathcal{G}(a, b, p, e)$.
    
    If $k_b\leq e-2$, the two cases of Eq.~\eqref{eq:T:k_a} hold.
    From the definition of $k_\alpha$, one has $p^{k_\alpha}\mid (x_0-\alpha)$ and $p^{k_\alpha+1}\nmid (x_0-\alpha)$. It yields from Eq.~\eqref{eq:p-adic} that $h_i=c_i$ for any $0\le i\le k_\alpha-1$, $h_{k_\alpha}\neq c_{k_\alpha}$, and $1\leq h_i\leq p-1$ for any $k_\alpha+1\le i\le e-1$.
    So, there are $(p-1)p^{e-k_\alpha-1}$ initial states of value $x_0$ satisfying $\logp(x_0-\alpha)=k_\alpha$.
    Then, there are $(p-1)p^{e-k_\alpha-1}$ and 
    $\sum^{e-1}_{k_\alpha=e-k_b}(p-1)p^{e-k_\alpha-1}=(p-1)\frac{p^{k_b-1}(1-p^{-k_b})}{1-p^{-1}}=p^{k_b}-1$
    initial states of value $x_0$ that satisfy the first and second cases of Eq.~\eqref{eq:T:k_a}, respectively.
    Thus, one can know all states of value $x_0$ satisfying $(x_0-\alpha)(x_0-\beta)\in(p)$
    compose $2\cdot \frac{(p-1)p^{e-k_\alpha-1}}{2p^{e-k_b-k_\alpha}}=(p-1)p^{k_b-1}$ cycles of length $2p^{e-k_b-k'}$, 
    $p^{k_b}-1$ cycles of length two, and two self-loops in $\mathcal{G}(a, b, p, e)$. 
    \end{itemize}
\end{proof}

When $(a, b, p, e)=(6, 5, 5, 2)$, one has $k_b=1$ and $4a+b^2\equiv7^2\pmod{5^2}$.
As shown in Fig.~\ref{fig:cycle:p2}b), $\mathcal{G}(6, 5, 5, 2)$ is composed of connected component $G(1, 4)$, $\frac{5-3}{2}=1$ cycle of length $2\times5=10$, $5^{2-1}-1=4$ cycles of length two, and two self-loops. The result is consistent with Proposition~\ref{pro:fzpe:unit}.

\begin{Proposition}
\label{pro:abn0b0:fnot}
When $a\in\Zx{e}$, $b\in(p)$, and $4a+b^2$ is a non-quadratic residue modulo $p^e$, $\mathcal{G}(a, b, p, e)$ is composed of
connected component $G(1, p^{e-1}-1)$ and $\frac{(p-1)p^{k_b-1}}{2}$ cycles of length $2p^{e-k_b}$, where $k_b=\logp(b)$.
\end{Proposition}
\begin{proof}
Similar to Proposition~\ref{pro:fzpe:unit}, one can know all states whose values belong to $(p)$ compose connected component $G(1, p^{e-1}-1)$. 

Combining the condition of this proposition and $4a+b^2 \not\equiv0\pmod p$, one has $f(t)$ is basic irreducible from Lemma~\ref{lemma:ft:baisc}. 
It means ring $\Z{e}[t]/(f(t))$ is isomorphic to $\GR(p^e, 2)$. 
So, 
it from Eq.~\eqref{eq:GR:form} that $\alpha$ and $\beta$ can be expressed as 
$\alpha=\sum^{e-1}_{i=0}c'_ip^i$ and $\beta=\sum^{e-1}_{i=0}d'_ip^i$, 
where $c'_i, d'_i\in\{0, 1, \xi, \cdots, \xi^{p^{2}-2}\}$ and $\xi$ is an element of order $p^2-1$ in $\GR(p^e, 2)$. 
For any $x_0\in\Zx{e}$, since $\alpha, \beta \notin\Zx{e}$, one has
$x_0\not\equiv c'_0 \text{~or~} d'_0 \pmod p$, which yields
$(x_0-\alpha)(x_0-\beta)\not\in(p)$. 
It yields from Lemma~\ref{lemma:L:an0b0} that $T(x_0; a, b)=\ord(\alpha\beta^{-1})_{p^e}=2p^{e-k_b}$. 
In addition, there are $(p-1)p^{e-1}$ initial states of value $x_0$ such that $x_0\in\Zx{e}$, which makes up $\frac{(p-1)p^{e-1}}{2p^{e-k_b}}$ cycles of length $2p^{e-k_b}$ in $\mathcal{G}(a, b, p, e)$.
\end{proof}

When $(a, b, p, e)=(7, 5, 5, 2)$, $k_b=1$ and $4a+b^2=53$ is a non-quadratic residue modulo $25$.
As shown in Fig.~\ref{fig:cycle:p2} c), $\mathcal{G}(7, 5, 5, 2)$ is composed of connected component $G(1, 4)$ and $\frac{(5-1)\times5}{2\times5}=2$ cycles of length $2\times5 =10$, which is consistent with Proposition~\ref{pro:abn0b0:fnot}.

\subsection{Graph structure of $\var{IPRNG}$~\eqref{eq:inv:zpe} with $a, b\in\Zx{e}$}
\label{subsec:abno:zpe}

When $a, b\in\Zx{e}$, $\var{IPRNG}$~\eqref{eq:inv:zpe} does not compose a permutation over $\Zx{e}$. 
Concretely, some states whose values belong to $\Zx{e}$ locate on a transient branch in $\mathcal{G}(a, b, p, e)$ (See Fig.~\ref{fig:cycle:p2}d).
By analyzing the pre-period and period of the sequence generated by $\var{IPRNG}$~\eqref{eq:inv:zpe} from such nodes, the structure of the corresponding connected components can be disclosed.
Equation~\eqref{eq:ynroot2:pe} plays an important role in the process of periodic analysis and is related to $(\alpha-\beta)^{-1}$.
When $a, b\in\Zx{e}$, $\alpha-\beta$ may not be a unit, which is different from the case in the previous subsection.
From Eq.~\eqref{eq:ab:alphabeta}, one has $4a+b^2=(\alpha-\beta)^2$. It deduces $\alpha-\beta$ is a unit if and only if $4a+b^2\not\equiv 0\pmod p$. 
The analysis of this subsection is divided according to the condition.

\subsubsection{$4a+b^2\not\equiv 0\pmod p$}

In such case, $\alpha-\beta$ must be a unit and Eq.~\eqref{eq:ynroot2:pe} still holds for any $x_0\in\Zx{e}$. 
The permutation nature of $\var{IPRNG}$~\eqref{eq:inv:zpe} ensures sequence $S(x_0; a, b)$ with $x_0\in\Zx{e}$ does not contain any elements in ideal $(p)$ when $a\in\Zx{e}$ and $b\in(p)$. 
However, when $a, b\in\Zx{e}$, sequence $S(x_0; a, b)$ may contain an element in ideal $(p)$. 
Referring to the condition, Lemma~\ref{lemma:abn0L:zpe} gives the explicit expression of the period of the sequence. 
% Eq.~\eqref{eq:ynroot2:pe} 
\begin{lemma}\label{lemma:abn0L:zpe}
When $a, b\in\Zx{e}$ and $4a+b^2\not\equiv 0\pmod p$, the pre-period of sequence $S(x_0; a, b)$ is not zero if and only if $(\bar{x}_0-\bar{\alpha})(\bar{x}_0-\bar{\beta})^{-1}\in \bar{\Omega}.$
And its least period
\begin{equation}
\label{eq:L:an0bn0}
T(x_0; a, b)=
\begin{cases}
1        & \mbox{if } x_0\in\{\alpha, \beta\};\\
\ord(\bar{\alpha}\bar{\beta}^{-1})_p-1 & \mbox{if } (\bar{x}_0-\bar{\alpha})(\bar{x}_0-\bar{\beta})^{-1}\in \bar{\Omega};\\
\ord(\alpha\beta^{-1})_{p^{e-k'}} & \mbox{otherwise, }
\end{cases}
\end{equation}
where $x_0\in\Z{e}$, 
$k'=\max(k_{\alpha}, k_{\beta})$,
$k_{\alpha}=\logp(\alpha)$,
$k_{\beta}=\logp(\beta)$,
and 
\[\bar{\Omega}
=\{\bar{\alpha}\bar{\beta}^{-1}, (\bar{\alpha}\bar{\beta}^{-1})^2, \cdots, (\bar{\alpha}\bar{\beta}^{-1})^{\ord(\bar{\alpha}\bar{\beta}^{-1})_p-1}\}.
\]
\end{lemma}
\begin{proof}
Since $\bar{y}_n=y_n\bmod p$, one has
\begin{equation}\label{eq:yn:w1w2}
\bar{y}_n=(\bar{\alpha}-\bar{\beta})^{-1}((\bar{x}_0-\bar{\beta})\bar{\alpha}^n+(\bar{\alpha}-\bar{x}_0)\bar{\beta}^n)\bmod p. 
\end{equation}
from Eq.~\eqref{eq:ynroot2:pe}. 
It yields from congruence~\eqref{eq:yn:w1w2} that
$\bar{y}_n=0$ is equivalent to
\begin{equation}\label{eq:w1w2}
(\bar{\alpha}\bar{\beta}^{-1})^n\equiv (\bar{x}_0-\bar{\alpha})(\bar{x}_0-\bar{\beta})^{-1}\pmod p. 
\end{equation}
If $(\bar{x}_0-\bar{\alpha})(\bar{x}_0-\bar{\beta})^{-1}\in \bar{\Omega}$, 
there exists an integer $s$ such that $\bar{y}_s=0$ from the above congruence, namely $y_s\in(p)$.
It yields from Lemma~\ref{lemma:xnToyne}-ii) that $x_{s-1}\in(p)$.
Hence, $T(x_0; a, b)=T(b; a, b)$. 
Setting $x_0=b$, from Eq.~\eqref{eq:ab:alphabeta}, one has 
\begin{equation}\label{eq:bar:x0}
\bar{x}_0=\bar{b}=\bar{\alpha}+\bar{\beta}. 
\end{equation}
It yields from Eq.~\eqref{eq:w1w2} that 
$\bar{y}_n=0$ if and only if
$(\bar{\alpha}\bar{\beta}^{-1})^{n+1}\equiv 1\pmod p.$
So, $m=\ord(\bar{\alpha}\bar{\beta}^{-1})_p-1$ is the smallest integer satisfying $\bar{y}_m=0$, namely $y_m\in(p)$.
Then, $x_{m-1} \in(p)$ from Lemma~\ref{lemma:xnToyne}.
Thus, sequence $S(x_0; a, b)$ is ultimately periodic and 
\[
T(x_0; a, b)=\ord(\bar{\alpha}\bar{\beta}^{-1})_p-1. 
\]
If $(\bar{x}_0-\bar{\alpha})(\bar{x}_0-\bar{\beta})^{-1}\notin \bar{\Omega}$ and $\bar{x}_0\in \{ \bar{\alpha}, \bar{\beta}\}$, one has $y_n\notin (p)$ for any integer $n$ from congruence~\eqref{eq:yn:w1w2}.
Then similar to Lemma~\ref{lemma:L:an0b0}, 
one can prove sequence $S(x_0; a, b)$ is periodic and 
\[
T(x_0; a, b)=
\begin{cases}
1                                 & \mbox{if } x_0\in\{\alpha, \beta\}; \\
\ord(\alpha\beta^{-1})_{p^{e-k'}} & \mbox{otherwise.}
\end{cases}
\]
\end{proof}

From Fig.~\ref{fig:cycle:p2}d), one can see that there exists connected component $G(2, 4)$ in $\mathcal{G}(6, 8, 5, 2)$.
Such structure illustrates there exist some nodes, whose values belong to $\Zx{e}$, located on a transient branch, namely, there are some initial states of value $x_0\in\Zx{e}$ satisfying sequence $S(x_0; a, b)$ is ultimately periodic.
Lemma~\ref{lemma:B:trasient} gives a set of such initial states of value.
Moreover, all nodes whose values belong to the set are leaf-nodes in the functional graph of $\var{IPRNG}$~\eqref{eq:inv:zpe}, as shown in Lemma~\ref{lemma:B:leafNode}.

\begin{lemma}
\label{lemma:B:trasient}
When $a, b\in\Zx{e}$ and $4a+b^2\not\equiv 0\pmod p$, 
the pre-period and least period of sequence $S(x_0; a, b)$
are both equal to $k-1$, and $x_{k-1}=b$,
where $k=\ord(\bar{\alpha}\bar{\beta}^{-1})_p$
and $x_0$ belongs to
\begin{equation}\label{set:B}
\mathbf{B}=\{x \mid x \equiv b \pmod p, x\neq b, x\in\Z{e}\}. 
\end{equation}
\end{lemma}
\begin{proof}
For any $x_0\in \mathbf{B}$, one has $\bar{x}_0=\bar{\alpha}+\bar{\beta}$ from Eq. \eqref{eq:bar:x0}.
So, 
\[
(\bar{x}_0-\bar{\alpha})(\bar{x}_0-\bar{\beta})^{-1}=\bar{\beta}\bar{\alpha}^{-1}=(\bar{\alpha}\bar{\beta}^{-1})^{-1}\in \bar{\Omega}.
\]
It yields from Lemma~\ref{lemma:abn0L:zpe} that sequence $S(x_0; a, b)$ is ultimately periodic and its period is $k-1$.
According to congruence~\eqref{eq:w1w2} and the above relation, one can know $\bar{y}_n=0$ is equivalent to
\begin{equation}
(\bar{\alpha}\bar{\beta}^{-1})^{n+1}\equiv 1\pmod p.
\end{equation}
So, $k-1$ is the smallest integer satisfying $\bar{y}_{k-1}=0$, namely $y_{k-1}\in(p)$.
Referring to Lemma~\ref{lemma:xnToyne}-ii), one has $x_{k-2} \in(p)$.
It means $x_{k-1}=\phi^{k-1}(x_0)=b$. 
From Lemma~\ref{lemma:abn0L:zpe}, one has $S(b; a, b)$ is periodic and its period is $k-1$. 
Thus, one can get the pre-period of sequence $S(x_0; a, b)$ is $k-1$.
\end{proof}
 
\begin{lemma}
\label{lemma:B:leafNode}
When $a, b\in\Zx{e}$, the in-degree of any node in set~\eqref{set:B} in $\mathcal{G}(a, b, p, e)$ is zero. 
\end{lemma}
\begin{proof}
For any node of value $y\in \mathbf{B}$, assume its in-degree is not zero. 
Then there exists a node of value $x\in \Z{e}$ such that
\[
y=\phi(x)=
\begin{cases}
(ax^{-1}+b)\bmod {p^e} & \mbox{if } x\in\Zx{e};\\
b& \mbox{if } x \in (p). 
\end{cases}
\]
From set~\eqref{set:B} and $a\in\Zx{e}$, one has $x\in (p)$. 
% and It follows from $y\equiv b\pmod p$ and $a\in\Zx{e}$ that 
Thus $\phi(x)=b$, 
it contradicts $y\neq b$. 
So, the lemma is proved. 
\end{proof}

\begin{Proposition}
\label{pro:2root:unit:Zpe}
When $a, b\in\Zx{e}$ and $4a+b^2$ is a quadratic residue modulo $p^e$, $\mathcal{G}(a, b, p, e)$ is composed of connected component $G(k-1, p^{e-1}-1)$, 
$\frac{(p-1-k)p^{e-1}}{\ord(\alpha\beta^{-1})_{p^e}}$
cycles of length $\ord(\alpha\beta^{-1})_{p^e}$, 
$\frac{2(p-1)p^{e-k'-1}}{\ord(\alpha\beta^{-1})_{p^{e-k'}}}$ cycles of length $\ord(\alpha\beta^{-1})_{p^{e-k'}}$, and two self-loops, where $k=\ord(\bar{\alpha}\bar{\beta}^{-1})_p$ and $k'\in \{1, 2, \cdots, e-1\}$.
\end{Proposition}
\begin{proof}
From the known condition of this proposition, one has 
the two roots $\alpha, \beta\in\Zx{e}$. 
Depending on the relation between the initial state of value $x_0$ and the two roots, the proof of this proposition is divided into the 
following three cases:
\begin{itemize}
    \item $\bar{x}_0\in \{\bar{\alpha}, \bar{\beta}\}$: One has 
    $(x_0-\alpha)(x_0-\beta)\in(p)$. 
    Assume $x_0-\alpha\in(p)$ and $x_0-\beta\in\Zx{e}$. 
    When $x_0=\alpha$, $T(x_0; a, b)=1$ from Lemma~\ref{lemma:abn0L:zpe}.
    It means the node of value $\alpha$ is a self-loop in $\mathcal{G}(a, b, p, e)$. 
    When $x_0\neq \alpha$, it follows from Lemma~\ref{lemma:abn0L:zpe} that sequence $S(x_0; a, b)$ is periodic and 
    \[
    T(x_0; a, b)=\ord(\alpha\beta^{-1})_{p^{e-k_\alpha}}.
    \]
    From Eq.~\eqref{eq:p-adic}, one can know there are $(p-1)p^{e-k_\alpha-1}$ initial states of value $x_0$ satisfying $\logp(x_0)=k_\alpha$. 
    In addition, the case $x_0-\alpha\in\Zx{e}$ and $x_0-\beta\in(p)$ is similar and omitted. 
    Thus, all states of value $x_0$ satisfying $\bar{x}_0\in \{ \bar{\alpha}, \bar{\beta}\}$ make up $\frac{2(p-1)p^{e-k'-1}}{\ord(\alpha\beta^{-1})_{p^{e-k'}}}$ cycles of length $\ord(\alpha\beta^{-1})_{p^{e-k'}}$ and two self-loops in $\mathcal{G}(a, b, p, e)$. 
    
    \item $(\bar{x}_0-\bar{\alpha})(\bar{x}_0-\bar{\beta})^{-1}\notin \bar{\Omega}$: 
    It follows from Lemma~\ref{lemma:abn0L:zpe} that
    $S(x_0; a, b)$ is periodic and
    \[
    T(x_0; a, b)=\ord(\alpha\beta^{-1})_{p^e}. 
    \]
    Referring to $\bar{x}_0\notin\{\bar{\alpha}, \bar{\beta}\}$ and the definition of $\bar{\Omega}$, one can know there are $(p-2-(k-1))p^{e-1}$ initial states of value $x_0$ satisfying $(\bar{x}_0-\bar{\alpha})(\bar{x}_0-\bar{\beta})^{-1}\notin\bar{\Omega}$. 
    It means all states $x_0$ satisfying $(\bar{x}_0-\bar{\alpha})(\bar{x}_0-\bar{\beta})^{-1}\notin \bar{\Omega}$ make up
    $\frac{(p-1-k)p^{e-1}}{\ord(\alpha\beta^{-1})_{p^e}}$
    cycles of length $\ord(\alpha\beta^{-1})_{p^e}$ in $\mathcal{G}(a, b, p, e)$. 
    
    \item $(\bar{x}_0-\bar{\alpha})(\bar{x}_0-\bar{\beta})^{-1}\in\bar{\Omega}$: It follows from Lemma~\ref{lemma:abn0L:zpe} that sequence $S(x_0; a, b)$ is ultimately periodic and $T(x_0; a, b)=k-1$. 
    Let $\mathbf{B}=\{b_i\}^{|\mathbf{B}|}_{i=1}$. 
    It follows from Lemma~\ref{lemma:B:trasient} that
    \begin{multline*}
    S(b_i; a, b)=\{b_i, \phi(b_i), \phi^2(b_i), \cdots, \phi^{k-2}(b_i), \\
    b, \phi(b), \phi^2(b), \cdots, \phi^{k-2}(b), b, \phi(b), \cdots \}.
    \end{multline*}
    So, one can get 
    \begin{equation}\label{eq:joint:Sx0}
    \bigcap\limits_{b_i\in \mathbf{B}}S(b_i; a, b)=\{ b, 
    \phi(b), \phi^2(b), \cdots, \phi^{k-2}(b)\}.
    \end{equation}
    From set~\eqref{set:B}, one can know $|\mathbf{B}|=p^{e-1}-1$, which further deduces $|\bigcup_{b_i\in B}S(b_i; a, b)|= (k-1)p^{e-1}$.
    From Lemma~\ref{lemma:B:leafNode} and set~\eqref{eq:joint:Sx0}, one can know all states belonging to $\bigcup_{b_i\in B}S(b_i; a, b)$ compose a unilateral connected digraph $G(k-1, p^{e-1}-1)$.
    From the definition of $\bar{\Omega}$ and Eq.~\eqref{eq:p-adic}, one can know there are $(k-1)p^{e-1}$ initial states of value $x_0$ satisfying $(\bar{x}_0-\bar{\alpha})(\bar{x}_0-\bar{\beta})^{-1}\in\bar{\Omega}$. 
    Thus, these states compose connected component $G(k-1, p^{e-1}-1)$ in $\mathcal{G}(a, b, p, e)$. 
\end{itemize}
\end{proof}

When $(a, b, p, e)=(8, 8, 5, 2)$, one can calculate $4a+b^2\equiv 14^2 \pmod {p^2}$, and $\alpha=22$, $\beta=11$, $\beta^{-1}=16$. 
So, $k=\ord(22\times 16)_{5}=4$.
As shown in Fig.~\ref{fig:cycle:p2}e), $\mathcal{G}(8, 8, 5, 2)$ is composed of connected component $G(3, 4)$, $\frac{2\times(5-1)}{4}=2$ cycles of length four
and two self-loops, which is consistent with Proposition~\ref{pro:2root:unit:Zpe}.

\begin{Proposition}
\label{pro:2root:unit:noZpe}
When $a, b\in\Zx{e}$ and $4a+b^2$ is a non-quadratic residue modulo $p^e$, $\mathcal{G}(a, b, p, e)$ is composed of connected component $G(k-1, p^{e-1}-1)$, 
$\frac{(p-k+1)p^{e-1}}{\ord(\alpha\beta^{-1})_{p^e}}$
cycles of length $\ord(\alpha\beta^{-1})_{p^e}$, where $k=\ord(\bar{\alpha}\bar{\beta}^{-1})_p$. 
\end{Proposition}
\begin{proof}
Similar to Proposition~\ref{pro:abn0b0:fnot}, one can know $\alpha,\beta$ belong to Galois ring $\GR(p^e, 2)$ not $\Zx{e}$. 
Thus, $\bar{x}_0-\bar{\alpha}$ and $\bar{x}_0-\bar{\beta}$ are units for any $x_0\in\Zx{e}$, which means $(x_0-\alpha)(x_0-\beta)\not\in(p). $
Depending on whether $(\bar{x}_0-\bar{\alpha})(\bar{x}_0-\bar{\beta})^{-1}\in\bar{\Omega}$ exists, 
the proof is divided into the following two cases:
\begin{itemize}
    \item $(\bar{x}_0-\bar{\alpha})(\bar{x}_0-\bar{\beta})^{-1}\notin\bar{\Omega}$: It follows that $T(x_0; a, b)=\ord(\alpha\beta^{-1})_{p^e}.$ 
    From the definition of $\bar{\Omega}$ and Eq.~\eqref{eq:p-adic}, one can know there are $(p-(k-1))p^{e-1}$ initial states of value $x_0$ satisfying $(\bar{x}_0-\bar{\alpha})(\bar{x}_0-\bar{\beta})^{-1}\notin\bar{\Omega}$. 
    Thus, these states compose $\frac{(p-k+1)p^{e-1}}{\ord(\alpha\beta^{-1})_{p^e}}$
    cycles of length $\ord(\alpha\beta^{-1})_{p^e}$ in $\mathcal{G}(a, b, p, e)$.
    
    \item $(\bar{x}_0-\bar{\alpha})(\bar{x}_0-\bar{\beta})^{-1}\in\bar{\Omega}$: Similar to Proposition~\ref{pro:2root:unit:Zpe}, one has
    $T(x_0; a, b)=\ord(\bar{\alpha}\bar{\beta}^{-1})_p-1$ and
    all states of value $x_0$ satisfying $(\bar{x}_0-\bar{\alpha})(\bar{x}_0-\bar{\beta})^{-1}\in\bar{\Omega}$ compose connected component $G(k-1, p^{e-1}-1)$ in $\mathcal{G}(a, b, p, e)$. 
    \end{itemize}
\end{proof}

Note that when $4a+b^2$ is a non-quadratic residue modulo $p^e$, similar to Eq.~\eqref{eq:ord:gt},
one can get
$\ord(\alpha\beta^{-1})_{p^e}=\ord(g(t))_{p^e}$,
$\ord(\bar{\alpha}\bar{\beta}^{-1})_p=\ord(\hat{g}(t))_p$,  
where $g(t)=t^2 +((a^{-1}b^2+2)\bmod p^e)t+1\in\mathbb{Z}_{p^e}[t]$ and $\hat{g}(t)=t^2+((a^{-1}b^2+2)\bmod p)t+1\in\mathbb{Z}_p[t]$.

When $(a, b, p, e)=(6, 8, 5, 2)$, $4a+b^2=88$ is a non-quadratic residue modulo $25$, 
$g(t)=t^2+21t+1\in\mathbb{Z}_{5^2}[t]$.
Assume $\epsilon$ is a root of $g(t)$, one has $\epsilon^2\equiv 4\epsilon-1 \pmod {5^2}$. Enumerating the power of $\epsilon$, one can get $\epsilon^{15}\equiv 1\pmod {5^2}$.
It means $\ord(\epsilon)_{5^2}=\ord(g(t))_{5^2}=\ord(\alpha\beta^{-1})_{5^2}=15$.
Similarly, one can calculate $\hat{g}(t)=t^2+t+1$ and $k=\ord(\bar{\alpha}\bar{\beta}^{-1})_5=\ord(\hat{g}(t))_5=3$.
As shown in Fig.~\ref{fig:cycle:p2}d), $\mathcal{G}(6, 8, 5, 2)$ is composed of connected component $G(2, 4)$ and $\frac{(5-3+1)\times5}{15}=1$ cycles of length $15$, which is consistent with Proposition~\ref{pro:2root:unit:noZpe}.

\subsubsection{$4a+b^2\equiv0\pmod p$} 
\label{subsec:abn0:zero}

In such case, $\alpha$ and $\beta$ belong to $\Z{e}$ if $f(t)$ is reducible in $\Z{e}[t]$; extension ring $\Z{e}[t]/(f(t))$ otherwise. 
Referring to $4a+b^2\equiv0\pmod p$ and Lemma~\ref{lemma:ft:baisc}, one can know $f(t)$ is not a basic irreducible polynomial. 
It means $\Z{e}[t]/(f(t))$ is not a Galois ring, and its element can be expressed as $q_0+q_1\beta$ where $q_0, q_1\in\Z{e}$~\cite{Chen:cat:TIT2012}. 
From $4a+b^2\equiv0\pmod p$, one has
$\hat{f}(t)=t^2-\bar{b}t-\bar{a}=
(t-\bar{\alpha})(t-\bar{\beta})=(t-\omega)^2$.
So, $\bar{\alpha}=\bar{\beta}=\omega$ and
\begin{equation}
\label{eq:barab:omega}
\bar{a}=-\omega^2 \text{~and~} \bar{b}=2\omega.
\end{equation}
Setting $\alpha=q_0+q_1 \beta$, it yields from $\bar{\alpha}=\bar{\beta}=\omega$ that
$q_0=p\mu\text{~and~}q_1=1+p\nu, $
where $\mu, \nu\in \Z{e}[t]/(f(t))$. 
Then $\alpha=p\mu +(1+p\nu)\beta =\beta+px$, 
where $x\in \Z{e}[t]/(f(t))$.
Then, one can calculate $\alpha^n=\beta^n+\sum_{i=1}^{n}\binom{n}{i}(px)^i\beta^{n-i}$
and
\begin{equation}\label{eq:alphabeta:p}
\begin{aligned}
\sum_{i=0}^{n-1} \alpha^{n-1-i}\beta^i
        & =\frac{\alpha^n-\beta^n}{\alpha-\beta} \\
        & =\frac{\sum_{i=1}^{n}\binom{n}{i}(px)^{i}\beta^{n-i}}{px}\\
        & =\sum_{i=1}^{n}\binom{n}{i}(px)^{i-1}\beta^{n-i}. 
\end{aligned}
\end{equation}
Lemma~\ref{lemma:pj:m} describes $p^j\mid n$ if the above equation is congruent to zero modulo $p^s$, where $s$ is a positive integer.
Referring to Lemma~\ref{lemma:pj:m}, the pre-period and period of sequence $S(x_0; a, b)$ for any $x_0\in \Z{e}$ is shown in Lemma~\ref{lemma:delta0:T}.

\begin{lemma}\label{lemma:pj:m}
For any integer $n$, if congruence $\sum_{i=1}^{n}\binom{n}{i}(px)^{i-1}\beta^{n-i}\equiv 0\pmod {p^s}$ holds, one has $p^j\mid n$, 
where $s\ge1$ and $j\leq s$. 
\end{lemma}
\begin{proof}
This lemma is proved via mathematical induction on $j$. 
When $j=1$, one has
$\sum_{i=1}^{n}\binom{n}{i}(px)^{i-1}\beta^{n-i}\equiv n\beta^{n-1} \equiv 0 \pmod p$. So, $p\mid n$. 
Assume that this lemma holds for $j=k<s$, namely $p^k\mid n$. 
It means $p^{k+1} \mid \binom{n}{i} p^{i-1}$ for any $i\in\{2, \cdots, n-1\}$.
% $$p^{k+1} \bigg| \binom{n}{i} p^{i-1}=\frac{n(n-1)(n-2)\cdots (n-(i-1))\cdot p^{i-1}}{i!}$$
When $j=k+1\leq s$, from $p^{k+1} \mid \binom{n}{i} p^{i-1}$, one has 
$\sum_{i=1}^{n}\binom{n}{i}(px)^{i-1}\beta^{n-i}
\equiv n\beta^{n-1} \equiv 0\pmod {p^{k+1}}.$
% +\binom{n}{2}px\beta^{n-2} \bmod p^{k+1}=0\bmod p^{k+1}. 
So, $p^{k+1}\mid n$. 
The above induction completes the proof of the lemma.
\end{proof}

\begin{lemma}
\label{lemma:delta0:T}
When $a, b\in\Zx{e}$ and $4a+b^2\equiv0\pmod p$, the pre-period of sequence $S(x_0; a, b)$ is zero if and only if
$\bar{x}_0=\omega$.
And its period
\begin{equation}
\label{eq:L:delta0}
T(x_0; a, b)=
\begin{cases}
p-1  & \mbox{if } \bar{x}_0\neq\omega; \\
p^{e-e_0} & \mbox{if } \bar{x}_0=\omega,
\end{cases}
\end{equation}
where $e_0=\logp(x_0^2-bx_0-a)$ and $\omega=2^{-1}\bar{b}\bmod p$.
\end{lemma}
\begin{proof}
From the definition of sequence $\{y_n\}_{n\ge 0}$ and Eq.~\eqref{eq:barab:omega}, one has 
\[
\bar{y}_{n+2}=\bar{b}\bar{y}_{n+1}+\bar{a}\bar{y}_n=2\omega \bar{y}_{n+1}-\omega^2\bar{y}_n \bmod p, 
\]
where $y_0=1$ and $y_1=x_0$.
So, 
\begin{equation}
\label{eq:overly}
\bar{y}_n=\omega^n(n(\omega^{-1}\bar{x}_0-1)+1) \bmod p. 
\end{equation}
It means $\bar{y}_n\neq 0$ for any $n$ if $\bar{x}_0=\omega$. 
According to whether the condition holds, 
the proof is divided into the following two cases:
\begin{itemize}
    \item $\bar{x}_0\neq \omega$:
    From congruence~\eqref{eq:overly}, 
    one can know there exists an integer number $n$ such that $\bar{y}_n=0$, namely $y_n\in(p)$. It follows from Lemma~\ref{lemma:xnToyne}-ii) that $S(x_0; a, b)$ must contain some element in $(p)$. 
    So, $T(x_0; a, b)=T(b; a, b)$. 
    Setting $x_0=b$, one can get $\bar{x}_0=2\omega$ and $\bar{y}_n=(n+1)\omega^n \bmod p$ by Eq.~\eqref{eq:overly}. 
    Thus, $m=p-1$ is the smallest integer such that $\bar{y}_m=0$, namely $y_{m}\in(p)$. 
    It yields from Lemma~\ref{lemma:xnToyne}-ii) that $x_{m-1}\in(p)$,
    which deduces $x_m =\phi(x_{m-1})=b$.
    So, $S(x_0; a, b)$ is ultimately periodic and 
    $T(x_0; a, b)=T(b; a, b)=p-1$.
    
    \item $\bar{x}_0=\omega$: One has $\bar{y}_n \neq 0$ for any $n$ from Eq.~\eqref{eq:overly}, which means that $S(x_0; a, b)$ does not contain any element in $(p)$. 
    So, $T(x_0; a, b)$ is the smallest integer $n$ that satisfies the congruence 
    $x_n=y_{n+1}\cdot y^{-1}_n\equiv x_0 \pmod {p^e}$, that is 
    \begin{equation}
    \label{eq:ynpe:x0}
    y_{n+1}-y_nx_0=cp^e, 
    \end{equation}
    where $c$ is an integer. 
    From Eq.~\eqref{eq:ab:alphabeta}, one has $(\alpha-\beta)^2=4a+b^2\equiv0\pmod p$ and $\alpha-\beta$ is not a unit. 
    It means the general term~\eqref{eq:ynroot2:pe} of relation~\eqref{eq:yn:pe} does not hold. 
    However, Eq.~\eqref{eq:ynroot2:pe} can be transformed into
    \[
    (\alpha-\beta)y_n=(x_0-\beta)\alpha^n+(\alpha-x_0)\beta^n \bmod p^e. 
    \]
    Multiplying both sides of Eq.~\eqref{eq:ynpe:x0} by $\alpha-\beta$
    and combining the above equation, one can get $(\alpha-\beta)(y_{n+1}-y_nx_0)=(\beta^n-\alpha^n)(x_0-\alpha)(x_0-\beta) \bmod p^e=cp^e(\alpha-\beta)$. 
    It yields that $T(x_0; a, b)$ is equal to the smallest integer $n\geq1$ that satisfies the congruence
    \begin{multline}
    \label{eq:alphabeta}
    \frac{(\alpha^n-\beta^n)}{\alpha-\beta}(x_0-\alpha)(x_0-\beta)=
    \sum_{i=0}^{n-1}(\alpha^{n-1-i}\beta^i)
    %(\alpha^{n-1}+\alpha^{n-2}\beta+\cdots+\beta^{n-1})\\
    (x_0-\alpha)\\
    (x_0-\beta)\equiv 0 \pmod {p^e}.
    \end{multline}
    % Since $\alpha$ and $\beta$ are two roots of polynomial $f(t)$, one has 
    Note that $f(x_0)=x_0^2-bx_0-a=(x_0-\alpha)(x_0-\beta)$. 
    So, $e_0=\logp((x_0-\alpha)(x_0-\beta))$. 
    It yields from Eq.~\eqref{eq:alphabeta:p} that $T(x_0; a, b)$ is the smallest integer $n$ satisfying
    \begin{multline*}
    \sum_{i=0}^{n-1} \alpha^{n-1-i}\beta^i=\sum_{i=1}^{n}\binom{n}{i}(px)^{i-1}\beta^{n-i}\\
    \equiv 0 \pmod {p^{e-e_0}}.
    \end{multline*}
    On the one hand, it is easy to verify that the above congruence holds with $n=p^{e-e_0}$, which means $T(x_0; a, b) \mid p^{e-e_0}$. 
    On the other hand, from Lemma~\ref{lemma:pj:m} and the above congruence holds with $n=p^{e-e_0}$, one can get $p^{e-e_0}\mid T(x_0; a, b)$. 
    Thus, $T(x_0; a, b)=p^{e-e_0}$.
\end{itemize}
\end{proof}

When $\bar{x}_0=\omega$, one has $T(x_0; a, b)=p^{e-e_0}$ from Lemma~\ref{lemma:delta0:T}.
Setting $x_0=\omega+hp$, one can calculate 
$x_0^2-bx_0-a=f(x_0)=f(\omega+hp)=f(\omega)+hp(2\omega-b)+h^2p^2$.
From Eq.~\eqref{eq:barab:omega} and the above equation, one has $2\omega-b \equiv 0\pmod p$ and 
\begin{equation}\label{eq:f_omega}
x_0^2-bx_0-a\equiv 
\begin{cases}
0 \pmod p;\\
f(\omega) \pmod{p^2}.
\end{cases}
\end{equation}
When $f(\omega)\not\equiv 0 \pmod {p^2}$, one has $e_0=1$ from the above congruence. 
Moreover, it follows from Lemma~\ref{lemma:delta0:T} that the period of sequence $S(x_0; a, b)$ is maximum in this condition.
Thus, in this subsection, we only give the graph structure of $\var{IPRNG}$~\eqref{eq:inv:zpe} that can generate the maximum periodic sequence, as shown in Proposition~\ref{pro:alphabeta:x=b}, and ignore other cases.

\begin{Proposition}
\label{pro:alphabeta:x=b}
When $4a+b^2\equiv 0\pmod p$ and $f(\omega)\not\equiv 0 \pmod {p^2}$, $\mathcal{G}(a, b, p, e)$ is composed of connected component $G(p-1, p^{e-1}-1)$ and one cycle of length $p^{e-1}$, where $\omega=2^{-1}\bar{b}\bmod p$.
\end{Proposition}
\begin{proof}
Similar to Lemma~\ref{lemma:delta0:T}, the proof is divided into the following two cases:
\begin{itemize}
    \item $\bar{x}_0=\omega$:
    From the definition of $e_0$ and congruence~\eqref{eq:f_omega}, one has $e_0=1$.
    It yields from Lemma~\ref{lemma:delta0:T} that $T(x_0; a, b)=p^{e-1}$.
    In addition, there are $p^{e-1}$ initial states of value $x_0$ satisfying $x_0=\omega$.
    Thus, these states make up one cycle of length $p^{e-1}$.
    
    \item $\bar{x}_0\neq\omega$:
    From the definition of set~\eqref{set:B}, one has $\bar{b}_i=2\omega$ for any $b_i\in \mathbf{B}$, where $i\in\{1,2,\cdots, |\mathbf{B}|\}$.
    Then, similar to the proof of case $\bar{x}_0\neq\omega$ of Lemma~\ref{lemma:delta0:T}, one can prove the pre-period and period of $S(b_i; a, b)$ are $p-1$ and $\phi^{p-1}(b_i)=b$. 
    Namely, 
    \begin{multline*}
    S(b_i; a, b)=\{b_i, \phi(b_i), \phi^2(b_i), \cdots, \phi^{p-2}(b_i),\\
    b, \phi(b), \phi^2(b), \cdots, \phi^{p-2}(b), b, \phi(b), \cdots \}.
    \end{multline*}
    Note that there are $(p-1)p^{e-1}$ initial states of value $x_0$ satisfying $\bar{x}_0\neq\omega$. 
    Similar to the proof of structure of $G(k-1, p^{e-1}-1)$ in Proposition~\ref{pro:2root:unit:Zpe}, one can prove these states
    make up connected component $G(p-1, p^{e-1}-1)$ in $\mathcal{G}(a, b, p, e)$. 
\end{itemize}
\end{proof}

When $(a, b, p, e)=(6, 6, 5, 2)$, one has $4a+b^2=60\equiv0\pmod{5}$, $\omega=2^{-1}\times 1=3$, and $f(\omega) \bmod 5^2=10$.
As shown in Fig.~\ref{fig:cycle:p2}f), $\mathcal{G}(6, 6, 5, 2)$ is composed of connected component $G(4, 4)$ and one cycle of length five, which is consistent with Proposition~\ref{pro:alphabeta:x=b}.

\section{Conclusion}

This paper reveals the functional graph of $\text{IPRNG}$ over the Galois ring $\mathbb{Z}_e$ by enumerating distinct sequences with the same least period, generated through the iteration of the generator from every state in the domain. Under every condition on its parameters, the relation between the period of each sequence and the parameters and the initial state of $\text{IPRNG}$ is explicitly formulated using the generating polynomial and linear difference equations.
The efficacy of the counting method relies on the intricate pattern of the functional graph: there is only one unilateral connected digraph and some cycles with a very small number of different lengths. The paper demonstrates that a non-negligible number of cycles with small periods exist when the parameters are selected improperly.
The analysis approach provides a potential solution to the open problem proposed by Sol\'e \emph{et al.} in \cite{Sole:IPRG:EJC2009} regarding the period distribution of the sequence generated by iterating $\text{IPRNG}$ over Galois rings $\text{GR}(p^e, n)$. The graph structure over the ring $\mathbb{Z}_N$ deserves further investigation when $N = p_1^{e_1}p_2^{e_2}\cdots p_n^{e_n}$ is a general composite.

\bibliographystyle{IEEEtran_doi.bst}
\bibliography{IPRNG.bib}

\end{document}